\documentclass[twocolumn,aps,showpacs,showkeys,byrevtex,pra]{revtex4-1}%
\pdfoutput=1
\usepackage{amsfonts}
\usepackage{amsmath}
\usepackage{amssymb}
\usepackage{graphicx}%
\providecommand{\U}[1]{\protect\rule{.1in}{.1in}}
\newtheorem{theorem}{Theorem}

\newtheorem{lemma}[theorem]{Lemma}

\newtheorem{proposition}[theorem]{Proposition}

\newenvironment{proof}[1][Proof]{\noindent\textbf{#1.} }{\ \rule{0.5em}{0.5em}}
\begin{document}
\preprint{ }
\title{Quantum trade-off coding for bosonic communication}
\author{Mark M. Wilde and Patrick Hayden}
\affiliation{\textit{School of Computer Science, McGill University, Montreal, Qu\'{e}bec,
Canada H3A 2A7}}
\author{Saikat Guha}
\affiliation{\textit{Quantum Information Processing Group, Raytheon BBN
Technologies, Cambridge, Massachusetts, USA 02138}}

\begin{abstract}
The trade-off capacity region of a quantum channel characterizes the optimal
net rates at which a sender can communicate classical, quantum, and entangled
bits to a receiver by exploiting many independent uses of the channel, along
with the help of the same resources. Similarly, one can consider a trade-off
capacity region when the noiseless resources are public, private, and secret
key bits. In [Phys.~Rev.~Lett.~\textbf{108}, 140501 (2012)], we identified
these trade-off rate regions for the pure-loss bosonic channel and proved that
they are optimal provided that a long-standing minimum output entropy
conjecture is true. Additionally, we showed that the performance gains of a
trade-off coding strategy when compared to a time-sharing strategy can be
quite significant. In the present paper, we provide detailed derivations of the
results announced there, and we extend the application of these ideas to
thermal-noise and amplifying bosonic channels. We also derive a
\textquotedblleft rule of thumb\textquotedblright\ for trade-off coding, which
determines how to allocate photons in a coding strategy if a large mean photon
number is available at the channel input. Our results on the amplifying
bosonic channel also apply to the \textquotedblleft Unruh channel\textquotedblright%
\ considered in the context of relativistic quantum information theory.

\end{abstract}
\date{\today}
\startpage{1}
\endpage{10}
\maketitle

\section{Introduction}

One of the great scientific accomplishments of the last century was Shannon's
formulation of information theory and his establishment of its two fundamental
theorems \cite{bell1948shannon}. Shannon's first theorem states that the
entropy of an information source is the best rate at which it can be
compressed, in the limit where many copies of the source are available. His
second theorem states that the maximum mutual information of a classical
channel is the highest rate at which information can be transmitted error-free
over such a channel, again in a limit where many independent uses of the
channel are available.

Shannon's theory is certainly successful in a world obeying the laws of
classical physics, but as we know, quantum mechanics is necessary in order to
describe the true physical nature of many communication channels. In
particular, a quantum-mechanical model is especially important for the case of
optical communication over free space or fiber-optic channels (the name for a
simple model of this channel is the \textit{pure-loss bosonic channel}
\cite{WPGCRSL12}). As such, we should also revise Shannon's theory of
information in order to account for quantum-mechanical effects, and Holevo,
Schumacher, and Westmoreland (HSW) \cite{ieee1998holevo,SW97} were some of the
first to begin this effort by proving that a quantity now known as the Holevo
information is an achievable rate for classical communication over a quantum
channel. Revising Shannon's information theory is not merely a theoretical
curiosity---the promise of quantum information theory is that communication
rates can be boosted by doing so \cite{YO93,GGLMSY04}, and recent experiments
have improved the state of the art in approaching the limits on communication
given by quantum information theory \cite{CHDLG12}.

The task of communicating classical data is certainly important, but the
communication of quantum data could be just as important, given the advent of
quantum computation \cite{book2000mikeandike} and given the possibility that distributed quantum
computation might one day become a reality \cite{BKBFZW12}. Lloyd
\cite{L97}, Shor \cite{capacity2002shor}, and Devetak \cite{ieee2005dev}%
\ (LSD) gave increasingly rigorous proofs that a quantity known as the
coherent information of a quantum channel is an achievable rate for quantum
communication, after Schumacher and others identified that this quantity would
be relevant for quantum data transmission
\cite{PhysRevA.54.2629,PhysRevLett.80.5695,BKN98,BNS98}.

In future communication networks, it is likely that a sender and receiver will
not be using communication channels to transmit either classical data alone or quantum
data alone, but rather that the data being transmitted will be a mix
of these data types. Additionally, it could be that the sender and receiver might
share some entanglement before communication begins, and it is well known that
entanglement can boost transmission rates \cite{PhysRevLett.83.3081,BSST01}. A
simple strategy for simultaneously communicating classical and quantum data
would be for the sender and receiver to use the best HSW\ classical code for a
fraction of the time and to use the best LSD\ quantum code for the other
fraction of the time (this simple strategy is known as time sharing). 
Devetak and Shor, however, demonstrated that this is not the optimal strategy in general
and that a trade-off coding strategy can outperform a time-sharing strategy
\cite{PhysRevLett.83.3081}. In short, a trade-off coding strategy is one in
which the sender encodes classical information into the many different ways of
permuting quantum error-correcting codes. After obtaining the channel outputs,
the receiver first performs a measurement to identify which permutation the
sender employed, and as long as the total number of permutations is not too
large, it is possible to identify the permutation with arbitrarily high
probability (and thus recovering the classical data that the sender transmitted).
Since this measurement is successful, it causes a negligible disturbance to
the state \cite{itit1999winter}, and the receiver can then decode the quantum
information encoded in the quantum codes.

We showed recently \cite{WHG12}\ that the performance gains can be very significant
for a pure-loss bosonic channel when employing a trade-off coding strategy
rather than a time-sharing strategy. (In this case, the trade-off coding
strategy amounts to a power-sharing strategy, like those considered in
classical information theory \cite{book1991cover}). We briefly summarize the
main results of Ref.~\cite{WHG12}. Suppose that a sender and receiver are
allowed access to many independent uses of a pure-loss bosonic channel with
transmissivity parameter $\eta\in\left[  0,1\right]  $ (so that $\eta$ is the
average fraction of photons that make it to the receiver). Suppose further
that the sender is power-constrained to input $N_{S}$ photons to the channel
on average per channel use. Let $C$ be the net rate of classical
communication (in bits per channel use), $Q$ the net rate of quantum communication
 (in qubits per channel use), and $E$ the net rate
of entanglement generation (in ebits per channel use). If a given rate is positive, then it means the
protocol generates the corresponding resource at the given rate, whereas if a
given rate is negative, then the protocol consumes the corresponding resource
at the given rate. The first main result of Ref.~\cite{WHG12}\ is that the
following rate region is achievable:%
\begin{align}
C+2Q &  \leq g\left(  \lambda N_{S}\right)  +g\left(  \eta N_{S}\right)
-g\left(  \left(  1-\eta\right)  \lambda N_{S}\right)  ,\nonumber\\
Q+E &  \leq g\left(  \eta\lambda N_{S}\right)  -g\left(  \left(
1-\eta\right)  \lambda N_{S}\right)  ,\nonumber\\
C+Q+E &  \leq g\left(  \eta N_{S}\right)  -g\left(  \left(  1-\eta\right)
\lambda N_{S}\right)  ,
\end{align}
where $g\left(  x\right)  \equiv\left(  x+1\right)  \log_{2}\left(
x+1\right)  -x\log_{2}x$ is the entropy of a thermal distribution with mean
photon number $x$ and $\lambda\in\left[  0,1\right]  $ is a photon-number
sharing parameter, determining the fraction of photons the code dedicates to
quantum communication (so that $1-\lambda$ is the fraction of photons that the
code dedicates to classical communication). The full region is the union over
all of the three-faced polyhedra given by the above inequalities, as the
photon-number sharing parameter $\lambda$ increases from zero to one. This
region is optimal whenever $\eta\geq1/2$ provided that a longstanding minimum
output entropy conjecture is true \cite{GGLMS04,GHLM10,GSE07,G08}, which
remains unproven but towards which proving there has been some good recent progress
\cite{KS12a,KS12b}.

The second main result of Ref.~\cite{WHG12}\ applies to a different though
related setting. The sender and receiver are again given access to many
independent uses of a pure-loss bosonic channel, but this time, the trade-off
is between the rate $R$ of public classical communication, the rate $P$ of
private classical communication, and the rate $S$ of secret key. We showed
that the following region is achievable:%
\begin{align}
R+P &  \leq g\left(  \eta N_{S}\right)  ,\nonumber\\
P+S &  \leq g\left(  \lambda\eta N_{S}\right)  -g\left(  \lambda\left(
1-\eta\right)  N_{S}\right)  ,\nonumber\\
R+P+S &  \leq g\left(  \eta N_{S}\right)  -g\left(  \lambda\left(
1-\eta\right)  N_{S}\right)  ,
\end{align}
where the parameter $\lambda\in\left[  0,1\right]  $ is again a photon-number
sharing parameter, determining the fraction of photons the code dedicates
to private communication (so that $1-\lambda$ is the fraction of photons that
the code dedicates to public communication). Optimality of this region whenever $\eta\geq1/2$ is again
subject to the same longstanding minimum output entropy conjecture. This
result has a practical relevance for a quantum key distribution protocol
executed using a pure-loss bosonic channel and forward public classical communication only.

In this paper, we provide detailed derivations of the above results announced
in Ref.~\cite{WHG12}, and we extend the ideas there to
thermal-noise and amplifying bosonic channels.
Section~\ref{sec:review-trade-off} recasts known results on trade-off
coding~\cite{HW08ITIT,HW10,WH10,WH10a}\ into a form more suitable for the
bosonic setting. Section~\ref{sec:lossy-bosonic} proves that the regions in
(1) and (2) are achievable and that they are the capacity regions provided a
long-standing minimum-output entropy conjecture is true \footnote{The simplest
version of the minimum-output entropy conjecture is stated as follows. Suppose that a thermal state
is injected into one input port of a beamsplitter. What state should we inject into the other input
port to minimize the entropy at one of the output ports? The conjecture is that the state should
just be the vacuum or a coherent state, but this question has been open
for some time \cite{GSE07,G08,GGLMS04,GSE08,GHLM10}. See Ref.~\cite{KS12a,KS12b} for recent progress.
}.
Section~\ref{sec:lossy-bosonic} also gives expressions for special cases of
the regions in (1) and (2), demonstrates that trade-off coding always beats
time-sharing whenever $0<\eta<1$, rigorously establishes a \textquotedblleft
rules of thumb\textquotedblright\ for trade-off coding outlined in
Ref.~\cite{WHG12}, and considers the high photon number limit for some special
cases. Sections~\ref{sec:thermal} and~\ref{sec:amplifier} give achievable rate
regions for the thermal-noise and amplifier bosonic channels, respectively. Finally,
Section~\ref{sec:unruh}\ demonstrates that an encoding with a fixed mean
photon budget can beat the rates achieved with a single-excitation encoding
for the \textquotedblleft Unruh channel.\textquotedblright\ We summarize our
results in the conclusion and suggest interesting lines of inquiry for
future research.

\section{Review of Trade-off Regions}

\label{sec:review-trade-off}In recent work, Wilde and Hsieh have synthesized
all known single-sender single-receiver quantum and classical communication
protocols over a general noisy quantum channel into a three-way dynamic
capacity region that they call the \textquotedblleft
triple-trade-off\textquotedblright\ region~\cite{HW08ITIT,HW10,WH10}. They
consider a three-dimensional region whose points ($C,Q,E$) correspond to rates
of classical communication, quantum communication, and entanglement generation
(or consumption) respectively. Let us illustrate a few simple examples of such
triple trade-offs. In the \emph{qubit teleportation} protocol \cite{BBCJPW93},
communication (consumption) of two classical bits and communication
(consumption) of one ebit of shared entanglement is used to communicate
(generate) one noiseless qubit, giving a rate triple $(-2,1,-1)$, where a
minus sign indicates the consumption of a resource and a plus sign indicates
the generation of a resource. Similarly, the superdense coding protocol
consumes a noiseless qubit channel and an ebit to generate two classical bits
\cite{BBCJPW93}. It corresponds to the rate triple $(2,-1,-1)$. Another simple
protocol is entanglement distribution, which communicates one noiseless qubit
to generate one ebit of shared entanglement, thereby yielding the rate triple
$(0,-1,1)$. Wilde and Hsieh have shown that when a communication channel is
noisy, the above three unit-resource protocols, in conjunction with the
\emph{classically-enhanced father protocol}~\cite{HW08ITIT,W11}, can be used
to derive the ultimate $(C,Q,E)$ trade-off space for any noisy quantum
channel. The dynamic capacity region's formulas are regularized over multiple
channel uses, and hence the triple trade-off region may in general be
superadditive. Another dynamic capacity region that Wilde and Hsieh
characterize is the trade-off between private communication rate $P$
(generation or consumption), public communication rate $R$ (generation or
consumption) and secret-key rate $S$ (distribution or
consumption)~\cite{WH10a}. We should clarify that these public and private rates are for
forward communication (generation or consumption).

\subsection{Quantum Dynamic Region}

\begin{proposition}
\label{prop:CQE-discrete}The quantum dynamic capacity region of a quantum
channel $\mathcal{N}$ is the regularization of the union of regions of the
form:%
\begin{align}
C+2Q  &  \leq H\left(  \mathcal{N}\left(  \rho\right)  \right) \nonumber\\
&  \ \ \ \ \ \ \ \ +\sum_{x}p_{X}\left(  x\right)  \left[  H\left(  \rho
_{x}\right)  -H\left(  \mathcal{N}^{c}\left(  \rho_{x}\right)  \right)
\right]  ,\label{eq:triple-bosonic-help-1}\\
Q+E  &  \leq\sum_{x}p_{X}\left(  x\right)  \left[  H\left(  \mathcal{N}\left(
\rho_{x}\right)  \right)  -H\left(  \mathcal{N}^{c}\left(  \rho_{x}\right)
\right)  \right]  ,\label{eq:triple-bosonic-help-2}\\
C+Q+E  &  \leq H\left(  \mathcal{N}\left(  \rho\right)  \right)  -\sum
_{x}p_{X}\left(  x\right)  H\left(  \mathcal{N}^{c}\left(  \rho_{x}\right)
\right)  , \label{eq:triple-bosonic-help-3}%
\end{align}
where $H(\sigma)\equiv-\operatorname{Tr} ( \sigma\ln\sigma)$, $\left\{
p_{X}\left(  x\right)  ,\rho_{x}\right\}  $ is an ensemble with expected
density operator $\rho\equiv\sum_{x}p_{X}\left(  x\right)  \rho_{x}$ and
$\mathcal{N}^{c}$ is the channel complementary to the channel $\mathcal{N}$.
\end{proposition}

\begin{proof}
The above proposition is just a rephrasing of the results from
Refs.~\cite{HW08ITIT,HW10,WH10}. First recall that the quantum mutual information of a bipartite state
$\rho^{AB}$ is defined as follows:
$$
I(A;B)_{\rho} \equiv H(A)_{\rho} +  H(B)_{\rho} -  H(AB)_{\rho},
$$
the conditional entropy is defined as
$$
H(A|B)_{\rho} \equiv H(AB)_{\rho} - H(B)_{\rho},
$$
and the coherent information is
$$
I(A\rangle B)_{\rho} \equiv - H(A|B)_{\rho} .
$$
Ref.~\cite{WH10} derived that the quantum
dynamic capacity region is the regularization of the union of regions of the
form:%
\begin{align}
C+2Q &  \leq I\left(  AX;B\right)  _{\rho},\label{eq:general-CQE-1}\\
Q+E &  \leq I\left(  A\rangle BX\right)  _{\rho},\label{eq:general-CQE-2}\\
C+Q+E &  \leq I\left(  X;B\right)  _{\rho}+I\left(  A\rangle BX\right)
_{\rho},\label{eq:general-CQE-3}%
\end{align}
where the union is over all classical-quantum states $\rho^{XAB}$ of the
following form:%
\begin{equation}
\rho^{XAB}\equiv\sum_{x}p_{X}\left(  x\right)  \left\vert x\right\rangle
\left\langle x\right\vert ^{X}\otimes\mathcal{N}^{A^{\prime}\rightarrow
B}(\phi_{x}^{AA^{\prime}}),\label{eq:code-states}%
\end{equation}
and each $\phi_{x}^{AA^{\prime}}$ is a pure, bipartite
state~\cite{HW08ITIT,HW10,WH10}. The convention in the above formulas is that
a rate for a resource is positive when the communication protocol generates
that resource, and the rate for a resource is negative when the communication
protocol consumes that resource. The above characterization is the full triple
trade-off, including both positive and negative rates. A simple strategy for
achieving the above capacity region is to combine the \textquotedblleft
classically-enhanced father protocol\textquotedblright\ \cite{HW08ITIT,W11}
with teleportation \cite{BBCJPW93}, superdense coding \cite{BW92}, and
entanglement distribution. The classically-enhanced father protocol is a
coding strategy that can communicate both classical and quantum information
with the help of shared entanglement. It generalizes both of the trade-off
coding strategies from Refs.~\cite{DS03,S04}.

We should clarify that the converse theorem from Ref.~\cite{WH10} applies even for the case of an infinite-dimensional channel. Without loss of generality, the sender's systems storing the shared entanglement, classical and quantum data to be transmitted can be assumed to be finite-dimensional, as can the receiver's systems after the decoding operation. If not, the sender and receiver can isometrically transfer their information to finite dimensional systems with negligible loss of fidelity. This is because finite numbers of perfect cbits, qubits, and ebits occupy only finite-dimensional subspaces of an infinite-dimensional Hilbert space. The converse theorem still holds in the infinite-dimensional case because all of the proofs begin by reasoning about the amount of information shared between a reference system and the decoded systems, which are finite-dimensional by the above assumption. Thus, applying continuity of entropy (as is done as a first step in all of the proofs) is not problematic. The proofs then make use of quantum data processing from the outputs of the infinite-dimensional channel to the decoded systems, and it is well-known that quantum data processing follows from monotonicity of quantum relative entropy \cite{W11}, an inequality which is robust in the infinite-dimensional case \cite{U77}. (This is the essential ingredient behind the Yuen-Ozawa proof of the infinite-dimensional variation of the Holevo bound [6].) Thus, the converse theorem from Ref.~\cite{WH10} still holds for the infinite-dimensional case. (A similar statement applies to the converse theorem from Ref.~\cite{WH10a}, which we employ later on in Section~\ref{sec:private-dynamic-review}.)

Due to the particular form of the state in (\ref{eq:code-states}), we can
rewrite the inequalities in (\ref{eq:general-CQE-1}-\ref{eq:general-CQE-3}) as%
\begin{align}
C+2Q  &  \leq H\left(  A|X\right)  _{\rho}+H\left(  B\right)  _{\rho}-H\left(
E|X\right)  _{\rho},\label{eq:CQE-entropies-1}\\
Q+E  &  \leq H\left(  B|X\right)  _{\rho}-H\left(  E|X\right)  _{\rho
},\label{eq:CQE-entropies-2}\\
C+Q+E  &  \leq H\left(  B\right)  _{\rho}-H\left(  E|X\right)  _{\rho},
\label{eq:CQE-entropies-3}%
\end{align}
where the entropies are now with respect to the following classical-quantum
state%
\begin{equation}
\rho^{XABE}\equiv\sum_{x}p_{X}\left(  x\right)  \left\vert x\right\rangle
\left\langle x\right\vert ^{X}\otimes U_{\mathcal{N}}^{A^{\prime}\rightarrow
BE}(\phi_{x}^{AA^{\prime}}),\label{eq:aug-code-states}%
\end{equation}
and $U_{\mathcal{N}}^{A^{\prime}\rightarrow BE}$ is an isometric extension of
the channel $\mathcal{N}^{A^{\prime}\rightarrow B}$. We can also think of the
information quantities as being with respect to the following input ensemble
(isomorphic to the input classical-quantum state in~(\ref{eq:aug-code-states}%
)):%
\begin{equation}
\left\{  p_{X}\left(  x\right)  ,\phi_{x}^{AA^{\prime}}\right\}  .
\label{eq:code-ensemble}%
\end{equation}
Let $\rho_{x}\equiv\ $Tr$_{A}\left\{  \phi_{x}^{AA^{\prime}}\right\}  $ and
let $\rho\equiv\sum_{x}p_{X}\left(  x\right)  \rho_{x}$. We then obtain the
inequalities in the statement of the proposition by substituting into
(\ref{eq:CQE-entropies-1}-\ref{eq:CQE-entropies-3}). \end{proof}

Observe that it suffices to calculate just four entropies in order to
determine the achievable rates $(C,Q,E)$ associated to particular input
ensemble:%
\begin{align}
&  H\left(  \mathcal{N}\left(  \rho\right)  \right)
,\label{eq:triple-trade-off-entropies-1}\\
&  \sum_{x}p_{X}\left(  x\right)  H\left(  \rho_{x}\right)
,\label{eq:triple-trade-off-entropies-2}\\
&  \sum_{x}p_{X}\left(  x\right)  H\left(  \mathcal{N}\left(  \rho_{x}\right)
\right)  ,\label{eq:triple-trade-off-entropies-3}\\
&  \sum_{x}p_{X}\left(  x\right)  H\left(  \mathcal{N}^{c}\left(  \rho
_{x}\right)  \right)  . \label{eq:triple-trade-off-entropies-4}%
\end{align}

\subsection{Private Dynamic Region}

\label{sec:private-dynamic-review}The private dynamic capacity region of a
quantum channel captures the trade-off between public classical communication,
private classical communication, and secret key~\cite{WH10a}.

\begin{proposition}
\label{prop:RPS-dynamic} The private dynamic capacity region of a degradable
quantum channel $\mathcal{N}$ is the regularization of the union of regions of
the form:%
\begin{align}
R+P  &  \leq H\left(  \mathcal{N}\left(  \rho\right)  \right) \nonumber\\
&  \ \ \ \ \ \ \ \ -\sum_{x,y}p_{X}\left(  x\right)  p_{Y|X}\left(
y|x\right)  H\left(  \mathcal{N}\left(  \psi_{x,y}\right)  \right)
,\label{eq:private-dynamic-region-1}\\
P+S  &  \leq\sum_{x}p_{X}\left(  x\right)  \left[  H\left(  \mathcal{N}\left(
\rho_{x}\right)  \right)  -H\left(  \mathcal{N}^{c}\left(  \rho_{x}\right)
\right)  \right]  ,\\
R+P+S  &  \leq H\left(  \mathcal{N}\left(  \rho\right)  \right)  -\sum
_{x}p_{X}\left(  x\right)  H\left(  \mathcal{N}^{c}\left(  \rho_{x}\right)
\right)  , \label{eq:private-dynamic-region-3}%
\end{align}
where $\left\{  p_{X}\left(  x\right)  p_{Y|X}\left(  y|x\right)  ,\left\vert
\psi_{x,y}\right\rangle \right\}  $ is an ensemble of pure states, $\rho
_{x}\equiv\sum_{y}p_{Y|X}\left(  y|x\right)  \psi_{x,y}$, $\rho\equiv\sum
_{x}p_{X}\left(  x\right)  \rho_{x}$, and $\mathcal{N}^{c}$ is the channel
complementary to $\mathcal{N}$.
\end{proposition}

\begin{proof}
This proposition is a rephrasing of the results in Ref.~\cite{WH10a}. For a
degradable quantum channel $\mathcal{N}$, Lemma~6 of Ref.~\cite{WH10a}\ states
that the private dynamic region is as follows:%
\begin{align*}
R+P  &  \leq I\left(  YX;B\right)  _{\omega},\\
P+S  &  \leq H\left(  B|X\right)  _{\omega}-H\left(  E|X\right)  _{\omega},\\
R+P+S  &  \leq H\left(  B\right)  _{\omega}-H\left(  E|X\right)  _{\omega},
\end{align*}
where $R$ is the rate of public classical communication, $P$ is the rate of
private classical communication, $S$ is the rate of secret key generation or
consumption, the state $\omega^{XYBE}$ is a state of the following form:%
\begin{multline*}
\omega^{XYBE}\equiv\sum_{x,y}p_{X}\left(  x\right)  p_{Y|X}\left(  y|x\right)
\left\vert x\right\rangle \left\langle x\right\vert ^{X}\\
\otimes\left\vert y\right\rangle \left\langle y\right\vert ^{Y}\otimes
U_{\mathcal{N}}^{A\rightarrow BE}\left(  \psi_{x,y}^{A}\right)  ,
\end{multline*}
$U_{\mathcal{N}}^{A\rightarrow BE}$ is an isometric extension of the channel
$\mathcal{N}^{A\rightarrow B}$, and $\psi_{x,y}^{A}$ are pure states. The
register $X$ is associated to the public information of the code, and the
register $Y$ is associated to the code's private information. The method for
achieving the above capacity region is to combine the \textquotedblleft
publicly-enhanced private father\textquotedblright\ protocol with the one-time
pad, private-to-public communication, and secret key distribution~\cite{WH10a}%
. The \textquotedblleft publicly-enhanced private father\textquotedblright%
\ protocol exploits shared secret key and the channel to communicate public
and private classical information~\cite{HW09a}.

We can also think about this capacity region from the ensemble point of view.
Let the input ensemble be%
\[
\left\{  p_{X}\left(  x\right)  p_{Y|X}\left(  y|x\right)  ,\psi
_{x,y}\right\}  .
\]
Let $\rho_{x}\equiv\sum_{y}p_{Y|X}\left(  y|x\right)  \psi_{x,y}$ and
$\rho\equiv\sum_{x}p_{X}\left(  x\right)  \rho_{x}$. Then we get the
characterization in the statement of the proposition by substitution.
\end{proof}

The above proposition also gives an achievable rate region if we restrict the
input states to be pure (pure states are sufficient to achieve the boundary of
the region for degradable channels, but a general channel might need an
optimization over all mixed states). Observe that it suffices to consider the
following four entropies in order to calculate the private dynamic capacity
region:%
\begin{align}
&  H\left(  \mathcal{N}\left(  \rho\right)  \right)  ,\label{eq:RPS-entropy-1}%
\\
&  \sum_{x,y}p_{X}\left(  x\right)  p_{Y|X}\left(  y|x\right)  H\left(
\mathcal{N}\left(  \psi_{x,y}\right)  \right)  ,\label{eq:RPS-entropy-2}\\
&  \sum_{x}p_{X}\left(  x\right)  H\left(  \mathcal{N}\left(  \rho_{x}\right)
\right)  ,\label{eq:RPS-entropy-3}\\
&  \sum_{x}p_{X}\left(  x\right)  H\left(  \mathcal{N}^{c}\left(  \rho
_{x}\right)  \right)  . \label{eq:RPS-entropy-4}%
\end{align}

\section{Lossy Bosonic Channel}

\label{sec:lossy-bosonic}Consider now the case of a single-mode lossy bosonic
channel. The transformation that this channel induces on the input
annihilation operators is%
\begin{align}
\hat{a}  &  \rightarrow\sqrt{\eta}\ \hat{a}+\sqrt{1-\eta}\ \hat{e}%
,\label{eq:lossy-bosonic-channel-1}\\
\hat{e}  &  \rightarrow-\sqrt{1-\eta}\ \hat{a}+\sqrt{\eta}\ \hat{e},
\label{eq:lossy-bosonic-channel-2}%
\end{align}
where $\hat{a}$ is the input annihilation operator for the sender, $\hat{e}$
is the input annihilation operator for the environment, and $\eta$ is the
transmissivity of the channel. Let $\mathcal{N}$ denote the Kraus map induced
by this channel, and let $\mathcal{N}^{c}$ denote the complementary channel.
In the case where the environmental input is the vacuum state, the
complementary channel is just a lossy bosonic channel with transmissivity
$1-\eta$~\cite{GSE08}. We place a photon number constraint on the input mode
to the channel, such that the mean number of photons at the input cannot be
greater than $N_{S}$.

\subsection{Quantum Dynamic Achievable Rate Region}

The proof of the theorem below justifies achievability of the region in~(1) in
the main text.

\begin{theorem}
An achievable quantum dynamic region for a lossy bosonic channel with
transmissivity $\eta$ is the union of regions of the form:%
\begin{align}
C+2Q  &  \leq g\left(  \lambda N_{S}\right)  +g\left(  \eta N_{S}\right)
-g\left(  \left(  1-\eta\right)  \lambda N_{S}\right)
,\label{eq:bosonic-region-1}\\
Q+E  &  \leq g\left(  \eta\lambda N_{S}\right)  -g\left(  \left(
1-\eta\right)  \lambda N_{S}\right)  ,\label{eq:bosonic-region-2}\\
C+Q+E  &  \leq g\left(  \eta N_{S}\right)  -g\left(  \left(  1-\eta\right)
\lambda N_{S}\right)  , \label{eq:bosonic-region-3}%
\end{align}
where $\lambda\in\left[  0,1\right]  $ is a photon-number-sharing parameter
and $g\left(  N\right)  $ is the entropy of a thermal state with mean photon
number $N$.
\end{theorem}

\begin{proof}
We pick an input ensemble of the form in (\ref{eq:code-ensemble}) as follows,
from which we will generate random codes:%
\begin{equation}
\left\{  p_{\overline{\lambda}N_{S}}\left(  \alpha\right)  ,D^{A^{\prime}%
}\left(  \alpha\right)  |\psi_{\text{TMS}}\rangle^{AA^{\prime}}\right\}  .
\label{eq:bosonic-ensemble}%
\end{equation}
The distribution $p_{\overline{\lambda}N_{S}}\left(  \alpha\right)  $ is an
isotropic Gaussian distribution with variance $\overline{\lambda}N_{S}$:%
\begin{equation}
p_{\overline{\lambda}N_{S}}\left(  \alpha\right)  \equiv\frac{1}{\pi
\overline{\lambda}N_{S}}\exp\left\{  -\left\vert \alpha\right\vert
^{2}/\overline{\lambda}N_{S}\right\}  , \label{eq:gaussian-prior}%
\end{equation}
where $\overline{\lambda}\equiv1-\lambda$, $\lambda\in\left[  0,1\right]  $ is
a photon-number-sharing parameter, indicating how many photons to dedicate to
the quantum part of the code, while $\overline{\lambda}$ indicates how many
photons to dedicate to the classical part. In~(\ref{eq:bosonic-ensemble}),
$D^{A^{\prime}}\left(  \alpha\right)  $ is a displacement operator acting on
mode $A^{\prime}$, and $|\psi_{\text{TMS}}\rangle^{AA^{\prime}}$ is a two-mode
squeezed (TMS) vacuum state of the following form~\cite{GK04,BVL05}:%
\begin{equation}
|\psi_{\text{TMS}}\rangle^{AA^{\prime}}\equiv\sum_{n=0}^{\infty}\sqrt
{\frac{\left[  \lambda N_{S}\right]  ^{n}}{\left[  \lambda N_{S}+1\right]
^{n+1}}}\left\vert n\right\rangle ^{A}\left\vert n\right\rangle ^{A^{\prime}}.
\label{eq:two-mode-squeezed}%
\end{equation}
Let $\theta$ denote the state resulting from tracing over the mode $A$:%
\begin{align*}
\theta &  \equiv\text{Tr}_{A}\left\{  |\psi_{\text{TMS}}\rangle\langle
\psi_{\text{TMS}}|^{AA^{\prime}}\right\} \\
&  =\sum_{n=0}^{\infty}\frac{\left[  \lambda N_{S}\right]  ^{n}}{\left[
\lambda N_{S}+1\right]  ^{n+1}}\left\vert n\right\rangle \left\langle
n\right\vert ^{A^{\prime}}.
\end{align*}
Observe that the reduced state $\theta$ is a thermal state with mean photon
number $\lambda N_{S}$~\cite{GK04}. Let $\overline{\theta}$ denote the state
resulting from taking the expectation of the state $\theta$ over the choice of
$\alpha$ with the prior $p_{\overline{\lambda}N_{S}}\left(  \alpha\right)  $:%
\[
\overline{\theta}\equiv\int d\alpha\ p_{\overline{\lambda}N_{S}}\left(
\alpha\right)  \ D\left(  \alpha\right)  \theta D^{\dag}\left(  \alpha\right)
.
\]
The state $\overline{\theta}$ is just a thermal state with mean photon number
$N_{S}$ \cite{GK04}. Thus, the state input to the channel on average has a
mean photon number $N_{S}$, ensuring that we meet the photon number constraint
on the channel input.

It is worth mentioning the two extreme cases of the ensemble in
(\ref{eq:bosonic-ensemble}). When the photon-number-sharing parameter
$\lambda=0$, the ensemble reduces to an ensemble of coherent states with a zero-mean
Gaussian prior of variance $N_{S}$:%
\[
\left\{  p_{N_{S}}\left(  \alpha\right)  ,\left\vert 0\right\rangle
^{A}\otimes\left\vert \alpha\right\rangle ^{A^{\prime}}\right\}  .
\]
This ensemble achieves the classical capacity of the lossy bosonic
channel~\cite{GGLMSY04}. When the photon-number-sharing parameter $\lambda=1$,
the input state is always the two-mode
squeezed state in (\ref{eq:two-mode-squeezed}) with $\lambda=1$, from which random codes
are then constructed. This input
state achieves both the entanglement-assisted classical and quantum capacities
of the lossy bosonic channel\ \cite{BSST01,HW01,GLMS03,GLMS03a}\ and the
channel's quantum capacity \cite{GSE08}. For the latter statement about
quantum capacity, the result holds if the mean input photon number is
sufficiently high (so that $g\left(  \eta N_{S}\right)  -g\left(  \left(
1-\eta\right)  N_{S}\right)  \approx\log_2\left(  \eta\right)  -\log_2\left(
1-\eta\right)  $) and, otherwise, the statement depends a long-standing
minimum-output entropy conjecture \cite{GGLMS04,GSE08,GHLM10}.

We can now calculate the various entropies in
(\ref{eq:triple-trade-off-entropies-1}-\ref{eq:triple-trade-off-entropies-4}).
For our case, they are respectively as follows:%
\begin{align}
&  H\left(  \mathcal{N}\left(  \overline{\theta}\right)  \right)
,\label{eq:bosonic-CQE-entropies-1}\\
&  \int d\alpha\ p_{\overline{\lambda}N_{S}}\left(  \alpha\right)  \ H\left(
D\left(  \alpha\right)  \theta D^{\dag}\left(  \alpha\right)  \right)
,\label{eq:bosonic-CQE-entropies-2}\\
&  \int d\alpha\ p_{\overline{\lambda}N_{S}}\left(  \alpha\right)  \ H\left(
\mathcal{N}\left(  D\left(  \alpha\right)  \theta D^{\dag}\left(
\alpha\right)  \right)  \right)  ,\label{eq:bosonic-CQE-entropies-3}\\
&  \int d\alpha\ p_{\overline{\lambda}N_{S}}\left(  \alpha\right)  \ H\left(
\mathcal{N}^{c}\left(  D\left(  \alpha\right)  \theta D^{\dag}\left(
\alpha\right)  \right)  \right)  . \label{eq:bosonic-CQE-entropies-4}%
\end{align}
The above entropies are straightforward to calculate for the case of the lossy
bosonic channel. We proceed in the above order. The state $\mathcal{N}\left(
\overline{\theta}\right)  $ is a thermal state with mean photon number $\eta
N_{S}$ (the lossy bosonic channel attenuates the mean photon number of the
transmitted thermal state), and so its entropy is%
\[
H\left(  \mathcal{N}\left(  \overline{\theta}\right)  \right)  =g\left(  \eta
N_{S}\right)  ,
\]
Entropy is invariant under the application of a unitary transformation, and so
the second entropy is just%
\begin{align*}
&  \int d\alpha\ p_{\overline{\lambda}N_{S}}\left(  \alpha\right)  \ H\left(
D\left(  \alpha\right)  \theta D^{\dag}\left(  \alpha\right)  \right) \\
&  =\int d\alpha\ p_{\overline{\lambda}N_{S}}\left(  \alpha\right)  \ H\left(
\theta\right) \\
&  =H\left(  \theta\right) \\
&  =g\left(  \lambda N_{S}\right)  .
\end{align*}
Both the channel $\mathcal{N}$ and the complementary channel $\mathcal{N}^{c}$
are covariant with respect to a displacement operator $D\left(  \alpha\right)
$ whenever the input state is thermal. Thus, we can compute the last two
entropies as%
\begin{align}
&  \int d\alpha\ p_{\overline{\lambda}N_{S}}\left(  \alpha\right)  \ H\left(
\mathcal{N}\left(  D\left(  \alpha\right)  \theta D^{\dag}\left(
\alpha\right)  \right)  \right) \nonumber\\
&  =\int d\alpha\ p_{\overline{\lambda}N_{S}}\left(  \alpha\right)  \ H\left(
D\left(  \sqrt{\eta}\alpha\right)  \mathcal{N}\left(  \theta\right)  D^{\dag
}\left(  \sqrt{\eta}\alpha\right)  \right) \nonumber\\
&  =H\left(  \mathcal{N}\left(  \theta\right)  \right) \nonumber\\
&  =g\left(  \eta\lambda N_{S}\right)  , \label{eq:CQE-Bob-entropy}%
\end{align}
and%
\begin{align}
&  \int d\alpha\ p_{\overline{\lambda}N_{S}}\left(  \alpha\right)  \ H\left(
\mathcal{N}^{c}\left(  D\left(  \alpha\right)  \theta D^{\dag}\left(
\alpha\right)  \right)  \right) \nonumber\\
&  =\int d\alpha\ p_{\overline{\lambda}N_{S}}\left(  \alpha\right)  \ H\left(
D\left(  \sqrt{\overline{\eta}}\alpha\right)  \mathcal{N}^{c}\left(
\theta\right)  D^{\dag}\left(  \sqrt{\overline{\eta}}\alpha\right)  \right)
\nonumber\\
&  =H\left(  \mathcal{N}^{c}\left(  \theta\right)  \right) \nonumber\\
&  =g\left(  \left(  1-\eta\right)  \lambda N_{S}\right)  ,
\label{eq:CQE-Eve-entropy}%
\end{align}
where $\overline{\eta}\equiv1-\eta$.

We can now specify our characterization of an achievable rate region for
trade-off communication over the lossy bosonic channel, simply by plugging in
our various entropies into the characterization of the region
in~(\ref{eq:triple-bosonic-help-1}-\ref{eq:triple-bosonic-help-3}). We can
justify this approach by means of a limiting argument similar to that which
appears in Refs.~\cite{GSE07,GSE08}. Suppose that we truncate the Hilbert
space at the channel input so that it is spanned by the Fock number states
$\left\{  \left\vert 0\right\rangle ,\left\vert 1\right\rangle ,\ldots
,\left\vert K\right\rangle \right\}  $ where $K\gg N_{S}$. Thus, all coherent
states, squeezed states, and thermal states become truncated to this
finite-dimensional Hilbert space. Also, it is only necessary to consider an
alphabet $\mathcal{X}$ that is finite because the input Hilbert space is
finite (this follows from Caratheodory's theorem~\cite{DS03,HW08ITIT}).
Applying Proposition~\ref{prop:CQE-discrete} to the ensemble of the form in
(\ref{eq:bosonic-ensemble}) in this truncated Hilbert space gives a quantum
dynamic region which is strictly an inner bound to the region in
(\ref{eq:bosonic-region-1}-\ref{eq:bosonic-region-3}). As we let $K$ grow
without bound, the entropies given by Proposition~\ref{prop:CQE-discrete}
converge to the entropies in (\ref{eq:bosonic-region-1}%
-\ref{eq:bosonic-region-3}). A similar argument applies to the private dynamic
regions and the other regions throughout this paper.
\end{proof}

In order to obtain the full region, we take the union of all the above regions
by varying the photon-number-sharing parameter $\lambda$ from $0$ to $1$. When
$\lambda$ is one, we are dedicating all of the photons to quantum resources
without the intent of sending any classical information. In the other case
when $\lambda$ is equal to zero, we are dedicating all of the photons to
sending classical information without any intent to send quantum information.
Figure~\ref{fig:CQE-bosonic}\ plots this triple trade-off region for
$\eta=3/4$ and $N_{S}=100$.%
\begin{figure}
[ptb]
\begin{center}
\includegraphics[
natheight=4.373400in,
natwidth=7.052600in,
width=3.339in
]%
{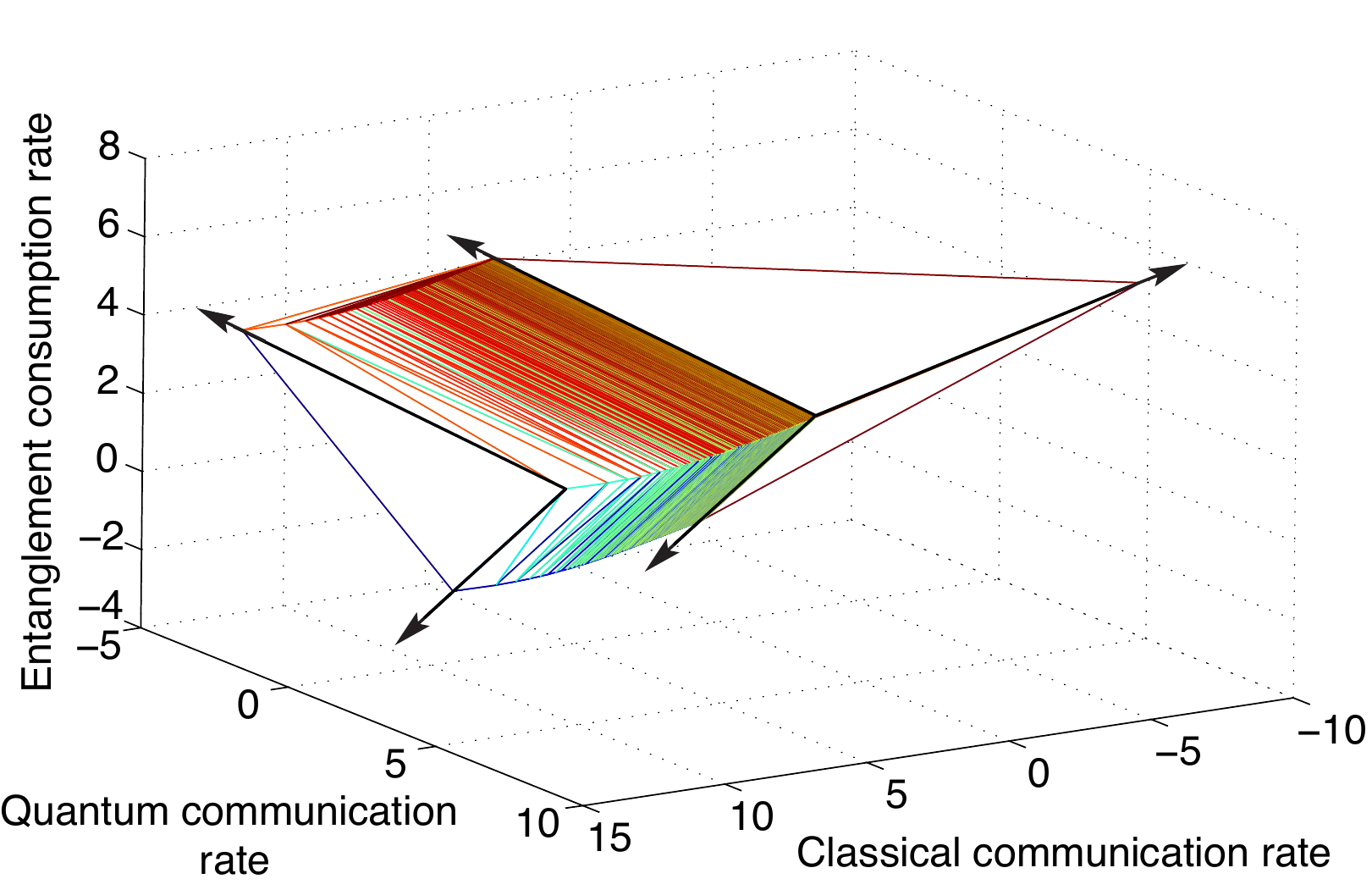}%
\caption{(Color online) The full triple trade-off region for the lossy bosonic channel with
transmissivity $\eta=3/4$ and mean input photon number $N_{S}=100$. Units for classical
communication, quantum communication, and entanglement consumption are bits per channel use,
qubits per channel use, and ebits per channel use, respectively.}%
\label{fig:CQE-bosonic}%
\end{center}
\end{figure}

\subsection{Special Cases of the Quantum Dynamic Trade-off}

\label{sec:cq-trade-off}A first special case of the quantum dynamic achievable
rate region is the trade-off between classical and quantum communication
(first explored by Devetak and Shor in Ref.~\cite{DS03}). The region in
(\ref{eq:bosonic-region-1}-\ref{eq:bosonic-region-3})\ reduces to the
following set of inequalities for this special case:%
\begin{align}
Q  &  \leq g\left(  \eta\lambda N_{S}\right)  -g\left(  \left(  1-\eta\right)
\lambda N_{S}\right)  ,\label{eq:CQ-1}\\
C+Q  &  \leq g\left(  \eta N_{S}\right)  -g\left(  \left(  1-\eta\right)
\lambda N_{S}\right)  . \label{eq:CQ-2}%
\end{align}
Figure~2(a)\ in the main text displays a plot of this region for a lossy
bosonic channel with transmissivity$~\eta=3/4$ and mean input photon
number$~N_{S}=200$. Trade-off coding for this channel can give a dramatic
improvement over a time-sharing strategy. More generally,
Figure~\ref{fig:lossy-bosonic-vary-eta}(a)\ demonstrates that trade-off coding
beats time-sharing for all$~\eta$ such that$~1/2<\eta<1$ (there is no
trade-off for$~\eta\leq1/2$ because the quantum capacity vanishes for these
values of$~\eta$).

Another special case of the quantum dynamic capacity region is when the sender
and receiver share prior entanglement and the sender would like to transmit
classical information to the receiver\ (a trade-off first explored by Shor in
Ref.~\cite{S04}). The region in (\ref{eq:bosonic-region-1}%
-\ref{eq:bosonic-region-3})\ reduces to the following set of inequalities for
this special case:%
\begin{align}
C  &  \leq g\left(  \lambda N_{S}\right)  +g\left(  \eta N_{S}\right)
-g\left(  \left(  1-\eta\right)  \lambda N_{S}\right)  ,\\
C  &  \leq g\left(  \eta N_{S}\right)  -g\left(  \left(  1-\eta\right)
\lambda N_{S}\right)  +E,
\end{align}
where we now take the convention that positive $E$ corresponds to the
consumption of shared entanglement. Figure~2(b)\ in the main text displays a
plot of this region for the case where $\eta=3/4$ and $N_{S}=200$. The figure
demonstrates that trade-off coding can give a dramatic improvement over
time-sharing. More generally, Figure~\ref{fig:lossy-bosonic-vary-eta}%
(b)\ demonstrates that trade-off coding beats time-sharing for all $\eta$ such
that $0<\eta<1$.%
\begin{figure*}
[ptb]
\begin{center}
\includegraphics[
natheight=5.066900in,
natwidth=11.899800in,
height=2.9974in,
width=6.7525in
]%
{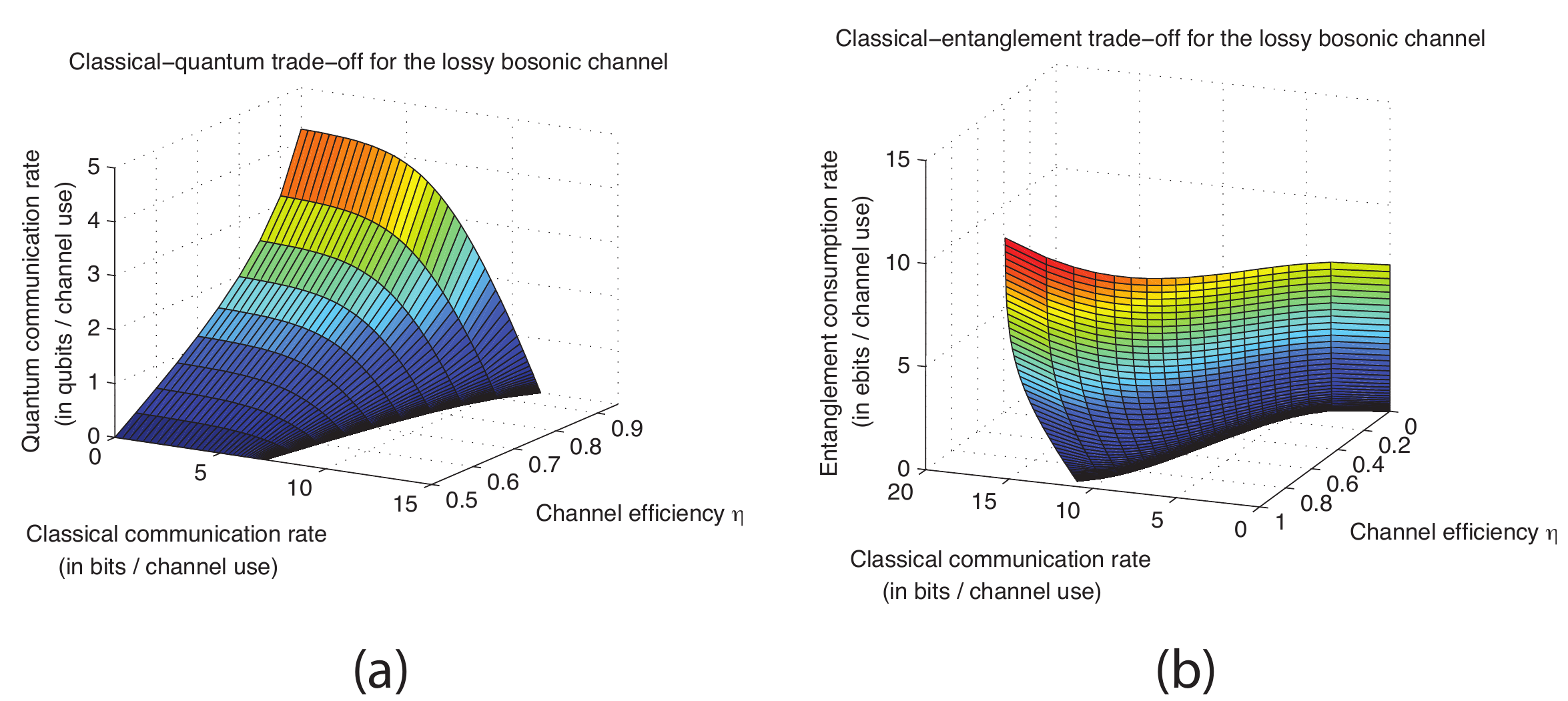}%
\caption{(Color online) (a) The trade-off between classical and quantum communication for all
$\eta\in\left(  1/2,1\right)  $, with the result that trade-off coding always
beats time-sharing by a significant margin. We assume that the mean input
photon number$~N_{S}=50/\left(  1-\eta\right)  $ so that there are a
sufficient number of photons to reach the quantum capacity of $\ln\left(
\eta\right)  -\ln\left(  1-\eta\right)  $ if the trade-off code dedicates all
of the available photons to quantum communication. (b) The trade-off between
entanglement-assisted and unassisted classical communication for all $\eta
\in\left(  0,1\right)  $, with the result that trade-off coding always beats
time-sharing. For consistency with (a), we again assume that $N_{S}=50/\left(
1-\eta\right)  $.}%
\label{fig:lossy-bosonic-vary-eta}%
\end{center}
\end{figure*}

\subsection{The Limit of High Mean Input Photon Number}

\label{sec:wall}We now briefly describe what happens to the above special
cases when the mean input photon number becomes high. We begin with the
classical-quantum trade-off. Recall from Ref.~\cite{HW01}\ that the classical
capacity of the lossy bosonic channel can be infinite if the input photon
number is unlimited. But the quantum capacity is fundamentally limited even if
an infinite number of photons are available~\cite{HW01,WPG07}. The threshold
on quantum capacity for the lossy bosonic channel with transmissivity $\eta$
is
\begin{equation}
\lim_{N_{S}\rightarrow\infty}g\left(  \eta N_{S}\right)  -g\left(  \left(
1-\eta\right)  N_{S}\right)  =\ln\left(  \eta\right)  -\ln\left(
1-\eta\right)  . \label{eq:q-cap-limit}%
\end{equation}
Let $N_{T}$ denote the approximate number of photons needed to reach the above
limit on quantum capacity. This limit has implications for trade-off coding.
Given a large input photon number $N_{S}>N_{T}$, it is possible to exploit an
amount $N_{T}$ for the quantum part of the transmission and $N_{S}-N_{T}$ for
the classical part of the transmission, so that the classical part of the
transmission can become arbitrarily large in the limit of infinite mean input
photon number (this leads to our rule of thumb pointed out in the main text).
Figure~\ref{fig:lossy-bosonic-vary-NS}(a) depicts the classical-quantum
trade-off for a lossy bosonic channel with$~\eta=3/4$ as the mean input photon
number~$N_{S}$ increases on a logarithmic scale from$~0.01$ to$~10^{10}$.%
\begin{figure*}
[ptb]
\begin{center}
\includegraphics[
natheight=4.046500in,
natwidth=8.947400in,
width=6.7525in
]%
{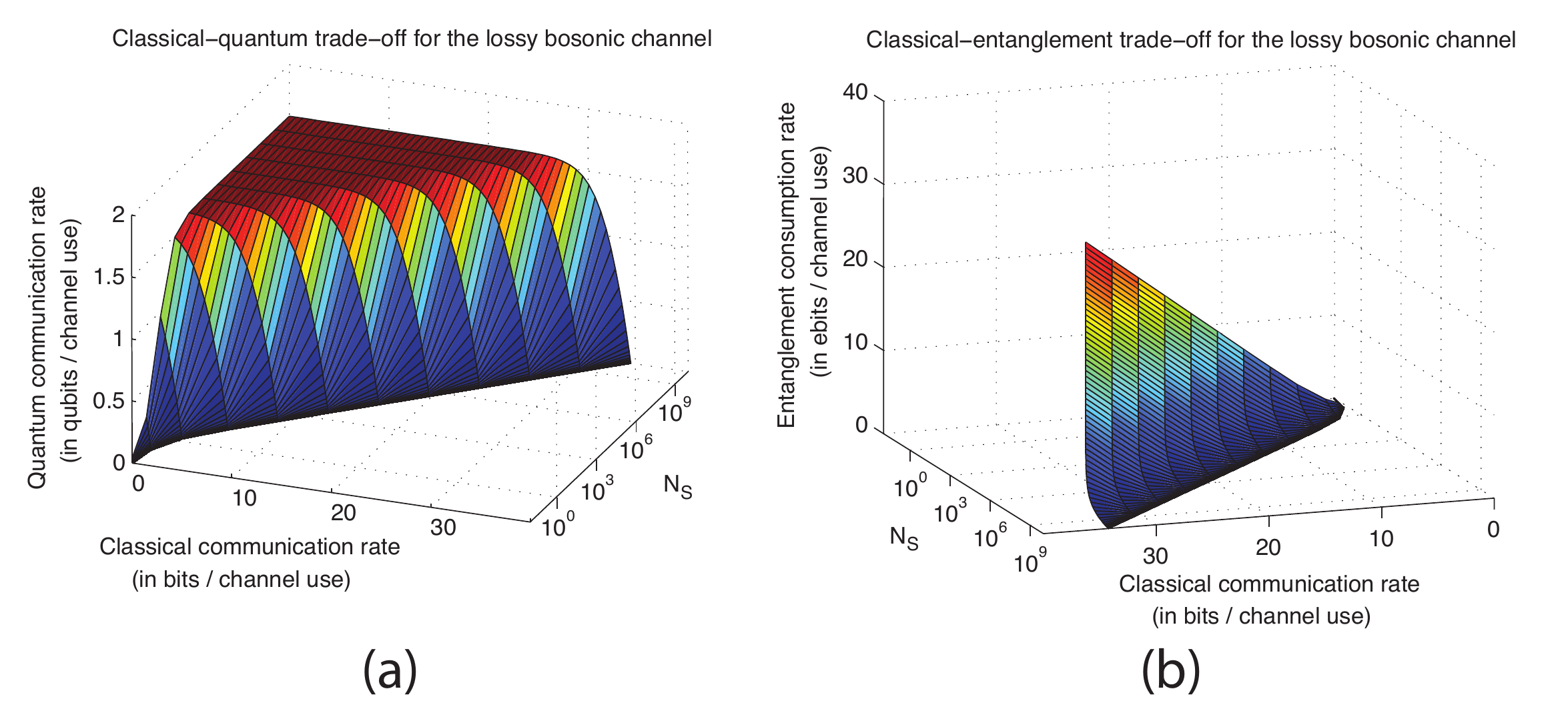}%
\caption{(Color online) (a) The trade-off between classical and quantum communication for a
lossy bosonic channel with transmissivity~$\eta=3/4$ as the mean input photon
number$~N_{S}$ increases on a logarithmic scale from~$0.01$ to~$10^{10}$. The
quantum capacity can never be larger than $\ln\left(  \eta\right)  -\ln\left(
1-\eta\right)  $, but the classical capacity is unbounded as $N_{S}%
\rightarrow\infty$. Thus, the best option for a trade-off code is to abide by
the rule of thumb given in the main text and dedicate only a small fraction
$\lambda\approx50/\left[  \left(  1-\eta\right)  N_{S}\right]  $ of the
available photons to quantum resources. This ensures that the quantum data
rate is near maximal while also maximizing the rate of classical
communication. (b) The trade-off between entanglement-assisted and unassisted
classical communication for the same lossy bosonic channel as $N_{S}$
increases logarithmically from~$0.01$ to~$10^{10}$. Abiding by a similar rule
of thumb and dedicating a fraction $\lambda\approx50/\left[  \left(
1-\eta\right)  N_{S}\right]  $ of the available photons to shared entanglement
ensures that the classical data rate is near maximal while minimizing the
entanglement consumption rate. Observe that trade-off protocols operating at
these near maximal data rates given by the rule of thumb consume about the
same amount of entanglement for all of the trade-off curves depicted.}%
\label{fig:lossy-bosonic-vary-NS}%
\end{center}
\end{figure*}

An interesting effect occurs with the trade-off between assisted and
unassisted classical communication. In the limit of infinite photon number,
the difference between the entanglement-assisted classical capacity and the
classical capacity approaches $\ln\left(  1/\left(  1-\eta\right)  \right)  $:%
\begin{align*}
&  \lim_{N_{S}\rightarrow\infty}g\left(  N_{S}\right)  +g\left(  \eta
N_{S}\right)  -g\left(  \left(  1-\eta\right)  N_{S}\right)  -g\left(  \eta
N_{S}\right) \\
&  =\lim_{N_{S}\rightarrow\infty}g\left(  N_{S}\right)  -g\left(  \left(
1-\eta\right)  N_{S}\right) \\
&  =\lim_{N_{S}\rightarrow\infty}\ln\left(  N_{S}\right)  -\ln\left(  \left(
1-\eta\right)  N_{S}\right) \\
&  =\ln\left(  1/\left(  1-\eta\right)  \right)  .
\end{align*}
Since both capacities diverge, this implies, in particular, that their ratio
approaches one in the same limit. This is an indication that entanglement
becomes less useful in the high photon number limit.
Figure~\ref{fig:lossy-bosonic-vary-NS}(b) depicts the classical-entanglement
trade-off for a lossy bosonic channel with transmissivity $\eta=3/4$ and mean
input photon number~$N_{S}$ increasing on a logarithmic scale from$~0.01$
to$~10^{10}$.

\subsection{``Rules of Thumb'' for Trade-off Coding}

Here we derive two propositions to support the claims in the main text
regarding two different ``rules of thumb'' for trade-off coding.

\begin{lemma}
\label{lem:Taylor}The thermal entropy function $g\left(  N\right)  $ admits
the following power series expansion whenever $N\geq1$:%
\[
g\left(  N\right)  =\ln\left(  N\right)  +1+\sum_{j=1}^{\infty}\frac{\left(
-1\right)  ^{j+1}}{j\left(  j+1\right)  N^{j}}.
\]

\end{lemma}

\begin{proof}
Consider the following chain of equalities:%
\begin{align}
g\left(  N\right)   &  =\left(  N+1\right)  \ln\left(  N+1\right)
-N\ln\left(  N\right) \nonumber\\
&  =\left(  N+1\right)  \ln\left(  1+\frac{1}{N}\right)  -1+\ln\left(
N\right)  +1 \label{eq:taylor-1}%
\end{align}
Now consider the following Taylor series expansion of $\ln\left(
1+1/N\right)  $ which is valid for all $N\geq1$:%
\begin{equation}
\ln\left(  1+\frac{1}{N}\right)  =\sum_{j=1}^{\infty}\frac{\left(  -1\right)
^{j+1}}{jN^{j}}. \label{eq:taylor-3}%
\end{equation}
We can use this expansion to manipulate the expression $\left(  N+1\right)
\ln\left(  1+1/N\right)  -1$:%
\begin{align}
&  \left(  N+1\right)  \left(  \sum_{j=1}^{\infty}\frac{\left(  -1\right)
^{j+1}}{jN^{j}}\right)  -1 =\sum_{j=1}^{\infty}\frac{\left(  -1\right)
^{j+1}}{j\left(  j+1\right)  N^{j}}. \label{eq:taylor-2}%
\end{align}
Combining (\ref{eq:taylor-1}), (\ref{eq:taylor-3}), and (\ref{eq:taylor-2})
gives the statement of the lemma.
\end{proof}

\begin{proposition}
A lower bound on the achievable rate for the lossy bosonic channel is as
follows:%
\[
g\left(  \eta N_{S}\right)  -g\left(  \left(  1-\eta\right)  N_{S}\right)
\geq\ln\left(  \eta\right)  -\ln\left(  1-\eta\right)  -\frac{1}{\eta\left(
1-\eta\right)  N_{S}}.
\]
provided that $\left(  1-\eta\right)  N_{S}\geq2$ and $\eta\geq1-\eta$. Thus,
in a trade-off coding strategy, it suffices to choose the photon-number
sharing parameter $\lambda=1/\left[  \eta\left(  1-\eta\right)  \epsilon
N_{S}\ln2\right]  $ whenever $\lambda\left(  1-\eta\right)  N_{S}\geq2$ and
$\eta\geq1-\eta$ in order to be within $\epsilon$ bits of the quantum capacity
(expressed in units of qubits per channel use).
\end{proposition}

\begin{proof}
The first step is to use the expansion from Lemma~\ref{lem:Taylor}.
Substituting in gives%
\begin{align*}
g(\eta N_{S})  &  -g((1-\eta)N_{S})\\
&  \geq\ln\eta-\ln(1-\eta)+\frac{1}{2\eta N_{S}}-\frac{1}{6\eta^{2}N_{S}^{2}%
}+\cdots\\
&  \;\quad-\frac{1}{2(1-\eta)N_{S}}+\frac{1}{6(1-\eta)^{2}N_{S}^{2}}-\cdots\\
&  \geq\ln\eta-\ln(1-\eta)-\frac{2\eta-1}{2\eta(1-\eta)N_{S}}\\
&  \ \ \ \ \ \ -\sum_{j=2}^{\infty}\frac{1}{j(j+1)}\frac{1}{[(1-\eta
)N_{S}]^{j}},
\end{align*}
where in the last line the order $1/N_{S}$ term is exact. The error term was
estimated by keeping only the negative terms in the expansion for orders
$1/N_{S}^{2}$ and higher, in addition to applying the inequality $1-\eta
\leq\eta$. Let $R=(1-\eta)N_{S}$. We can then bound%
\[
\sum_{j=2}^{\infty}\frac{1}{j(j+1)R^{j}}\leq\frac{1}{6}\sum_{j=2}^{\infty
}\frac{1}{R^{j}}=\frac{1}{6(R^{2}-R)}%
\]
using the formula for the sum of a geometric series. So if $R=(1-\eta
)N_{S}\geq2$, then $R\leq R^{2}/2$ and we get that
\begin{align*}
&  g(\eta N_{S})-g((1-\eta)N_{S})\\
&  \geq\ln\eta-\ln(1-\eta)-\frac{2\eta-1}{2\eta(1-\eta)N_{S}}-\frac{1}{3}%
\frac{1}{(1-\eta)N_{S}}\\
&  =\ln\eta-\ln(1-\eta)-\frac{8\eta/3-1}{2\eta(1-\eta)N_{S}}\\
&  \geq\ln\eta-\ln(1-\eta)-\frac{1}{\eta(1-\eta)N_{S}},
\end{align*}
using in the last line that $\eta<1$.

The statement in the proposition about trade-off coding follows by analyzing
the above bound. An achievable rate for quantum data transmission with a
trade-off coding strategy is $g\left(  \lambda\eta N_{S}\right)  -g\left(
\lambda\left(  1-\eta\right)  N_{S}\right)  $, and the above development gives
the following lower bound on this achievable rate (in units of qubits per
channel use):%
\[
\log_{2}\left(  \eta\right)  -\log_{2}\left(  1-\eta\right)  -\frac{1}%
{\eta\left(  1-\eta\right)  \lambda N_{S}\ln2}.
\]
Thus, if we would like to be within $\epsilon$ bits of the maximum quantum
capacity $\log_{2}\left(  \eta\right)  -\log_{2}\left(  1-\eta\right)  $, then
it suffices to choose the photon-number sharing parameter $\lambda$ as given
in the statement of the proposition.
%such that%
%\begin{align*}
%\epsilon &  \geq\frac{1}{\eta\left(  1-\eta\right)  \lambda N_{S}\ln2}\\
%\therefore\lambda &  \geq\frac{1}{\eta\left(  1-\eta\right)  \epsilon N_{S}%
%\ln2}.
%\end{align*}
\end{proof}

\begin{proposition}
An upper bound on the difference between the entanglement-assisted classical
capacity with maximal entanglement and that for limited entanglement trade-off
coding is as follows:%
\[
\frac{5}{6\lambda N_{S}\left(  1-\eta\right)  },
\]
provided that $\lambda\left(  1-\eta\right)  N_{S}\geq2$. Thus, in order to be
within $\epsilon$ bits of the entanglement-assisted classical capacity, it
suffices to choose the photon-number sharing parameter $\lambda=5/\left[
6\epsilon N_{S}\left(  1-\eta\right)  \ln2\right]  $ whenever $\lambda\left(
1-\eta\right)  N_{S}\geq2$.
\end{proposition}

\begin{proof}
Consider the difference between the entanglement-assisted classical capacity
$g\left(  N_{S}\right)  +g\left(  \eta N_{S}\right)  -g\left(  \left(
1-\eta\right)  N_{S}\right)  $ and the limited entanglement classical data
rate $g\left(  \lambda N_{S}\right)  +g\left(  \eta N_{S}\right)  -g\left(
\left(  1-\eta\right)  \lambda N_{S}\right)  $:%
\begin{align*}
&  g\left(  N_{S}\right)  +g\left(  \eta N_{S}\right)  -g\left(  \left(
1-\eta\right)  N_{S}\right) \\
&  \ \ \ \ \ \ -\left[  g\left(  \lambda N_{S}\right)  +g\left(  \eta
N_{S}\right)  -g\left(  \left(  1-\eta\right)  \lambda N_{S}\right)  \right]
\\
&  =g\left(  N_{S}\right)  -g\left(  \lambda N_{S}\right) \\
&  \ \ \ \ \ \ -\left[  g\left(  \left(  1-\eta\right)  N_{S}\right)
-g\left(  \left(  1-\eta\right)  \lambda N_{S}\right)  \right] \\
&  =\sum_{j=1}^{\infty}\frac{\left(  -1\right)  ^{j+1}}{j\left(  j+1\right)
}\left[  \frac{1}{N_{S}^{j}}-\frac{1}{\left(  \lambda N_{S}\right)  ^{j}%
}\right] \\
&  \ \ \ \ \ -\sum_{j=1}^{\infty}\frac{\left(  -1\right)  ^{j+1}}{j\left(
j+1\right)  N_{S}^{j}}\left[  \frac{1}{\left(  1-\eta\right)  ^{j}}-\frac
{1}{\left(  \lambda\left(  1-\eta\right)  \right)  ^{j}}\right]
\end{align*}
\begin{align*}
&  =\frac{1}{2N_{S}}\left[  1-\lambda^{-1}-\left(  1-\eta\right)
^{-1}+\left[  \lambda\left(  1-\eta\right)  \right]  ^{-1}\right] \\
&  \ \ \ \ \ \ +\sum_{j=2}^{\infty}\frac{\left(  -1\right)  ^{j+1}}{j\left(
j+1\right)  }\left[  \frac{1}{N_{S}^{j}}-\frac{1}{\left(  \lambda
N_{S}\right)  ^{j}}\right] \\
&  \ \ \ \ \ \ -\sum_{j=2}^{\infty}\frac{\left(  -1\right)  ^{j+1}}{j\left(
j+1\right)  N_{S}^{j}}\left[  \frac{1}{\left(  1-\eta\right)  ^{j}}-\frac
{1}{\left(  \lambda\left(  1-\eta\right)  \right)  ^{j}}\right]
\end{align*}
The second equality follows by expanding the thermal entropy functions with
Lemma~\ref{lem:Taylor} and by canceling terms. The third equality follows by
taking out the first term in the summation. Continuing,%
\begin{align*}
&  \leq\frac{1}{2N_{S}}\left(  \frac{\eta}{1-\eta}\right)  \left(
\frac{1-\lambda}{\lambda}\right) \\
&  +\sum_{j=2}^{\infty}\frac{2}{j\left(  j+1\right)  \left(  \lambda\left(
1-\eta\right)  N_{S}\right)  ^{j}}\\
&  \leq\frac{1}{2\lambda N_{S}\left(  1-\eta\right)  }+\sum_{j=2}^{\infty
}\frac{2}{j\left(  j+1\right)  \left(  \lambda\left(  1-\eta\right)
N_{S}\right)  ^{j}}.
\end{align*}
The first inequality follows by realizing that $1-\lambda^{-1}-\left(
1-\eta\right)  ^{-1}+\left[  \lambda\left(  1-\eta\right)  \right]
^{-1}=\left(  \eta/\left(  1-\eta\right)  \right)  \left(  1-\lambda\right)
/\lambda$, by keeping only the positive terms in the series, and by realizing
that $\lambda\left(  1-\eta\right)  N_{S}\leq\left(  1-\eta\right)  N_{S}$,
$\lambda\left(  1-\eta\right)  N_{S}\leq\lambda N_{S}$, and $\lambda\left(
1-\eta\right)  N_{S}\leq N_{S}$. The next inequality follows because
$\eta,1-\lambda\leq1$. Continuing,%
\begin{align*}
&  \leq\frac{1}{2\lambda N_{S}\left(  1-\eta\right)  }+\frac{2}{6}\sum
_{j=2}^{\infty}\frac{1}{\left(  \lambda\left(  1-\eta\right)  N_{S}\right)
^{j}}\\
&  =\frac{1}{2\lambda N_{S}\left(  1-\eta\right)  }+\frac{2}{6}\left(
\frac{1}{\left(  \lambda\left(  1-\eta\right)  N_{S}\right)  ^{2}%
-\lambda\left(  1-\eta\right)  N_{S}}\right)
\end{align*}
\begin{align*}
&  \leq\frac{1}{2\lambda N_{S}\left(  1-\eta\right)  }+\frac{2}{6}\frac
{1}{\lambda\left(  1-\eta\right)  N_{S}}\\
&  =\frac{5}{6\lambda N_{S}\left(  1-\eta\right)  }.
\end{align*}
The first inequality follows because $1/j\left(  j+1\right)  \leq1/6$ for
$j\geq3$. The first equality follows from the formula for a geometric series.
The second inequality is true whenever $\lambda\left(  1-\eta\right)
N_{S}\geq2$, and the final equality follows from simple addition.

The statement about trade-off coding follows from the above bound on the
difference between the entanglement-assisted classical capacity and the
limited entanglement data rate. If we would like the error to be within
$\epsilon~$bits of capacity, then the bound in the statement of the
proposition should hold.
%%
%\begin{align*}
%\epsilon &  \geq\frac{5}{6\lambda\left(  1-\eta\right)  N_{S}\ln2}\\
%\therefore\lambda &  \geq\frac{5}{6\epsilon\left(  1-\eta\right)  N_{S}\ln2}.
%\end{align*}
\end{proof}

\subsection{Private Dynamic Achievable Rate Region}

The proof of the theorem below justifies achievability of the region in~(3) in
the main text.

\begin{theorem}
An achievable private dynamic region for a lossy bosonic channel with
transmissivity $\eta$ is the union of regions of the form:%
\begin{align}
R+P  &  \leq g\left(  \eta N_{S}\right)  ,\label{eq:private-dynamic-1}\\
P+S  &  \leq g\left(  \lambda\eta N_{S}\right)  -g\left(  \lambda\left(
1-\eta\right)  N_{S}\right)  ,\label{eq:private-dynamic-2}\\
R+P+S  &  \leq g\left(  \eta N_{S}\right)  -g\left(  \lambda\left(
1-\eta\right)  N_{S}\right)  . \label{eq:private-dynamic-3}%
\end{align}
where $\lambda\in\left[  0,1\right]  $ is a photon-number-sharing parameter
and $g\left(  N\right)  $ is defined in (2).
\end{theorem}

\begin{proof}
This channel is degradable whenever $\eta\geq1/2$ \cite{GSE08}, and
antidegradable otherwise. For simplicity, we study the case where the channel
is degradable (Lemma~5 of Ref.~\cite{WH10a}\ demonstrates that the region is
somewhat trivial for the case of an antidegradable channel). We choose the
input ensemble to be a mixture of coherent states:%
\begin{equation}
\left\{  p_{\overline{\lambda}N_{S}}\left(  \alpha\right)  p_{\lambda N_{S}%
}\left(  \beta\right)  ,\left\vert \alpha+\beta\right\rangle \right\}  ,
\label{eq:ensemble-public-private}%
\end{equation}
where the distributions $p_{\overline{\lambda}N_{S}}\left(  \alpha\right)  $
and $p_{\lambda N_{S}}\left(  \beta\right)  $ are isotropic Gaussian priors of
the form in (\ref{eq:gaussian-prior}). The parameter $\lambda$ is again a
photon-number-sharing parameter where $\overline{\lambda}=1-\lambda$ is the
fraction of photons that the code dedicates to public resources and $\lambda$
is the fraction that it dedicates to private resources. Let$~\theta_{\alpha}$
denote the state resulting from averaging over the variable$~\beta$:%
\begin{align}
\theta_{\alpha}  &  \equiv\int d\beta\ p_{\lambda N_{S}}\left(  \beta\right)
\ \left\vert \alpha+\beta\right\rangle \left\langle \alpha+\beta\right\vert
\nonumber\\
&  =D\left(  \alpha\right)  \theta D^{\dag}\left(  \alpha\right)  ,
\label{eq:private-dynamic-conditional-state}%
\end{align}
where $\theta$ is a thermal state with mean photon number $\lambda N_{S}$. Let
$\overline{\theta}$ denote the state resulting from averaging over all states
in the ensemble:%
\begin{align}
\overline{\theta}  &  \equiv\int\int d\alpha\ d\beta\ p_{\overline{\lambda
}N_{S}}\left(  \alpha\right)  p_{\lambda N_{S}}\left(  \beta\right)
\ \left\vert \alpha+\beta\right\rangle \left\langle \alpha+\beta\right\vert
\nonumber\\
&  =\int d\alpha\ p_{\overline{\lambda}N_{S}}\left(  \alpha\right)  \ D\left(
\alpha\right)  \theta D^{\dag}\left(  \alpha\right)  .
\label{eq:private-dynamic-unconditional-state}%
\end{align}
Observe that $\overline{\theta}$ is just a thermal state with mean photon
number $N_{S}$, so that the mean number of photons entering the channel meets
the constraint of $N_{S}$.

We remark on the two extreme cases of the ensemble in
(\ref{eq:ensemble-public-private}). If the photon-number-sharing parameter
$\lambda=0$, then the coding scheme devotes all of its photons to public
classical communication. The ensemble is an isotropic distribution of coherent
states, which is the ensemble needed to achieve the capacity of the lossy
bosonic channel for public classical communication~\cite{GGLMSY04}. If the
photon-number-sharing parameter $\lambda=1$, then the coding scheme devotes
all of its photons to private classical communication. The ensemble is again
an isotropic mixture of coherent states, which is the ensemble needed to
achieve the private classical capacity of the lossy bosonic
channel~\cite{GSE08}, up to the aforementioned minimum-output entropy
conjecture~\cite{GGLMS04,GHLM10}.

The four entropies in (\ref{eq:RPS-entropy-1}-\ref{eq:RPS-entropy-4}) become
the following four entropies for our case:%
\begin{align}
&  H\left(  \mathcal{N}\left(  \overline{\theta}\right)  \right)  ,\\
&  \int\int d\alpha\ d\beta\ p_{\overline{\lambda}N_{S}}\left(  \alpha\right)
p_{\left(  \lambda\right)  N_{S}}\left(  \beta\right)  \ H\left(
\mathcal{N}\left(  \left\vert \alpha+\beta\right\rangle \left\langle
\alpha+\beta\right\vert \right)  \right)
,\label{eq:private-conditional-coherent-entropy}\\
&  \int d\alpha\ p_{\overline{\lambda}N_{S}}\left(  \alpha\right)  \ H\left(
\mathcal{N}\left(  \theta_{\alpha}\right)  \right)  ,\\
&  \int d\alpha\ p_{\overline{\lambda}N_{S}}\left(  \alpha\right)  \ H\left(
\mathcal{N}^{c}\left(  \theta_{\alpha}\right)  \right)  .
\end{align}
The first entropy is equal to the entropy of an attenuated thermal state:%
\[
H\left(  \mathcal{N}\left(  \overline{\theta}\right)  \right)  =g\left(  \eta
N_{S}\right)  .
\]
The second entropy is equal to zero because a lossy bosonic channel does not
change the purity of a coherent state. We calculate the final two entropies in
the same way as we did in (\ref{eq:CQE-Bob-entropy}) and
(\ref{eq:CQE-Eve-entropy}), respectively.
\end{proof}

Figure~\ref{fig:RPS-region}\ plots the private dynamic capacity region for a
lossy bosonic channel with transmissivity $\eta=3/4$ and the mean input photon
number $N_{S}=100$.%
\begin{figure}
[ptb]
\begin{center}
\includegraphics[
natheight=4.760800in,
natwidth=6.473100in,
width=3.4411in
]%
{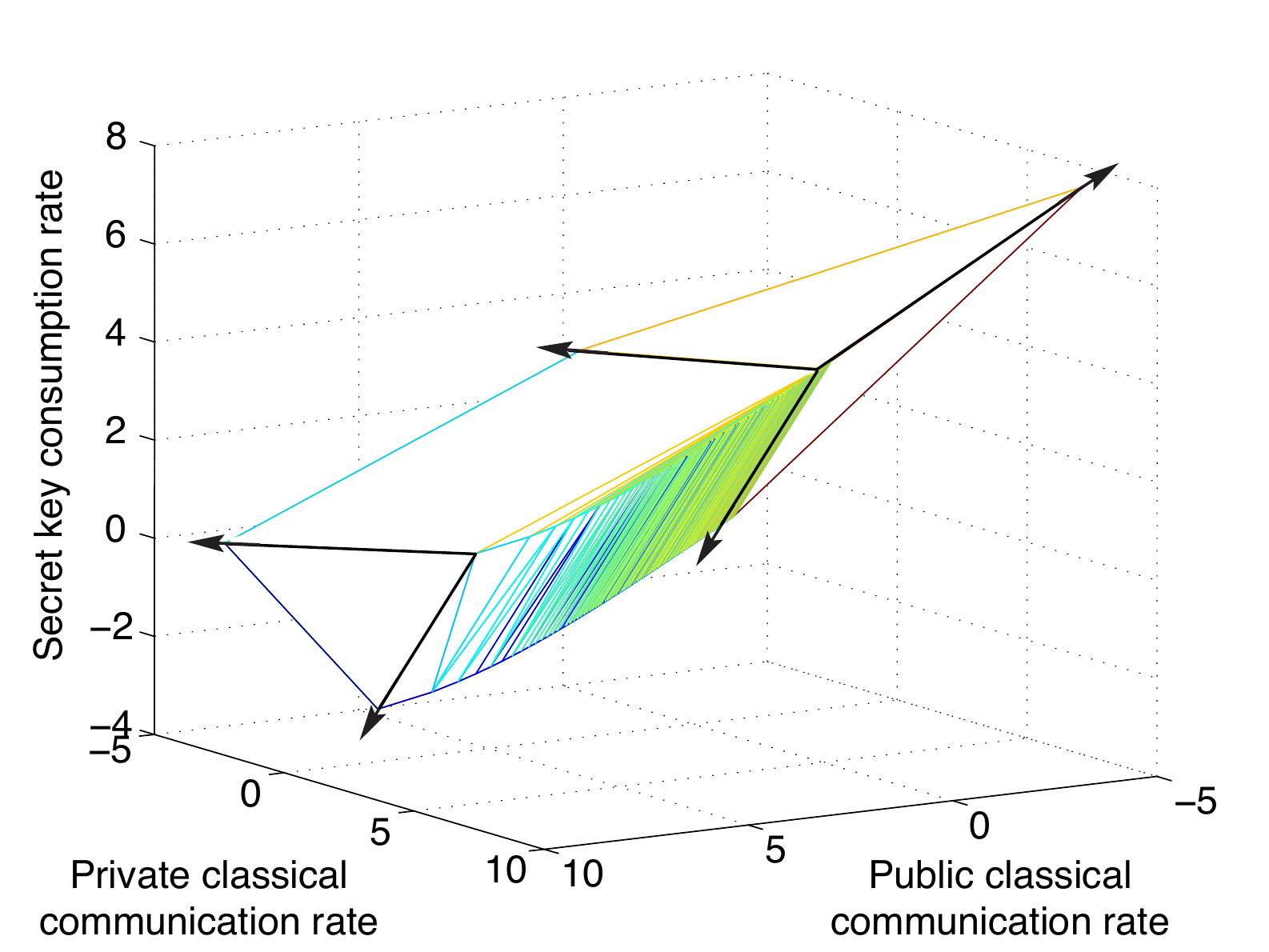}%
\caption{(Color online) The private dynamic capacity region for a lossy bosonic channel with
transmissivity $\eta=3/4$ and mean input photon number $N_{S}=100$. The units for each kind of resource are bits per channel use.} 
\label{fig:RPS-region}%
\end{center}
\end{figure}

\subsection{Special Case of the Private Dynamic Region}

An interesting special case of the private dynamic region is the trade-off
between public and private classical communication. Lemma~3 of
Ref.~\cite{WH10a} proves that the classical-quantum trade-off from
Appendix~\ref{sec:cq-trade-off}\ is the same as the public-private trade-off
whenever the channel is degradable (recall that the lossy bosonic channel is
degradable whenever $\eta\geq1/2$). Thus, the formulas in (\ref{eq:CQ-1}%
-\ref{eq:CQ-2}) characterize this trade-off as%
\begin{align}
P  &  \leq g\left(  \eta\lambda N_{S}\right)  -g\left(  \left(  1-\eta\right)
\lambda N_{S}\right)  ,\\
R+P  &  \leq g\left(  \eta N_{S}\right)  -g\left(  \left(  1-\eta\right)
\lambda N_{S}\right)  .
\end{align}
and Figure~2(a)\ serves as a plot of it. We should note that the trade-off
between public classical communication and secret key generation is the same
as that between public and private classical communication, found simply by
replacing $P$ with $S$ in the above formulas.

\subsection{Converse for the Quantum Dynamic Region}

We now prove that the characterization in (\ref{eq:bosonic-region-1}%
-\ref{eq:bosonic-region-3}) is the capacity region corresponding to the
trade-off between classical communication, quantum communication, and
entanglement. Our converse proof exploits the same ideas used in
Refs.~\cite{GSE07,G08} to prove optimality of the bosonic broadcast channel
region and, as such, it is optimal only if the minimum output entropy
conjecture is true (in particular, the second strong version from
Refs.~\cite{GSE07,G08}). Some evidence has been collected suggesting that this
conjecture should be true, but a full proof remains elusive.
Refs.~\cite{GSE07,G08}, however, have shown that the conjecture is true if the
input states are restricted to be Gaussian, and thus our region is optimal if
the input states are restricted to be Gaussian.

\begin{proof}
Proposition~\ref{prop:CQE-discrete} states that the regularization of the
region in (\ref{eq:triple-bosonic-help-1}-\ref{eq:triple-bosonic-help-3}) is
the quantum dynamic capacity region for any quantum channel. We prove here
that the region in (\ref{eq:bosonic-region-1}-\ref{eq:bosonic-region-3}) is
equivalent to the capacity region for a lossy bosonic channel $\mathcal{N}$
with transmissivity parameter $\eta>1/2$ and mean input photon number$~N_{S}$,
up to a minimum output entropy conjecture. We do so by proving the following
upper bounds:%
\begin{align}
\sum_{x}p_{X}\left(  x\right)  H\left(  \rho_{x}\right)   &  \leq ng\left(
\lambda N_{S}\right)  ,\label{eq:CQE-conv-1}\\
H\left(  \mathcal{N}^{\otimes n}\left(  \rho\right)  \right)   &  \leq
ng\left(  \eta N_{S}\right)  ,\label{eq:CQE-conv-2}\\
\sum_{x}p_{X}\left(  x\right)  H\left(  \mathcal{N}^{\otimes n}\left(
\rho_{x}\right)  \right)   &  \leq ng\left(  \eta\lambda N_{S}\right)  ,
\label{eq:CQE-conv-3}%
\end{align}
and the following lower bound:%
\begin{equation}
\sum_{x}p_{X}\left(  x\right)  H\left(  \left(  \mathcal{N}^{c}\right)
^{\otimes n}\left(  \rho_{x}\right)  \right)  \geq ng\left(  \left(
1-\eta\right)  \lambda N_{S}\right)  , \label{eq:CQE-conv-4}%
\end{equation}
so that for all $n$-letter ensembles $\left\{  p_{X}\left(  x\right)
,\rho_{x}\right\}  $ with $\rho_{x}\in\mathcal{B}\left(  \mathcal{H}^{\otimes
n}\right)  $ and $\rho\equiv\sum_{x}p_{X}\left(  x\right)  \rho_{x}$, there
exists some $\lambda\in\left[  0,1\right]  $ such that the above bounds hold.
The above bounds immediately imply that the region in
(\ref{eq:bosonic-region-1}-\ref{eq:bosonic-region-3}) is the quantum dynamic
capacity region.

The second bound in (\ref{eq:CQE-conv-2}) follows because the quantum entropy
is subadditive and the entropy of a thermal state gives the maximum entropy
for a bosonic state with mean photon number $\eta N_{S}$ (if the input mean
photon number is $N_{S}$, then the output mean photon number is $\eta N_{S}$
for a lossy bosonic channel with transmissivity~$\eta$).

Recall that the thermal entropy function $g\left(  x\right)  $ is
monotonically increasing and concave in its input argument. The proof from
Refs.~\cite{GSE07,G08} makes use of these facts and we can do so as well. Let
us begin by bounding the term in (\ref{eq:CQE-conv-1}). Supposing that the
mean number of photons for the $j^{\text{th}}$ symbol of $\rho_{x}$ is
$N_{S,x_{j}}$, we have the following bound:%
\begin{align}
0  &  \leq H\left(  \rho_{x}\right) \label{eq:CQE-conv-Alice-bound-1}\\
&  \leq\sum_{j=1}^{n}H\left(  \rho_{x}^{j}\right) \\
&  \leq\sum_{j=1}^{n}g\left(  N_{S,x_{j}}\right) \\
&  \leq ng\left(  N_{S,x}\right)  , \label{eq:CQE-conv-Alice-bound-4}%
\end{align}
where $N_{S,x}\equiv\sum_{j=1}^{n}\frac{1}{n}N_{S,x_{j}}$. The second
inequality exploits the subadditivity of quantum entropy, the third inequality
exploits the fact that the maximum quantum entropy of a bosonic system with
mean photon number~$N$ is $g\left(  N\right)  $, and the last inequality
exploits concavity of $g\left(  x\right)  $. Thus, for all $x\in\mathcal{X}$
there exists some $\lambda_{x}^{\prime}\in\left[  0,1\right]  $ such that%
\begin{equation}
H\left(  \rho_{x}\right)  =ng\left(  \lambda_{x}^{\prime}N_{S,x}\right)  ,
\label{eq:CQE-Alice-ind-ent}%
\end{equation}
because $g\left(  x\right)  $ is a monotonically increasing function of $x$
for $x\geq0$. Also, we have that%
\begin{align}
0  &  \leq\sum_{x}p_{X}\left(  x\right)  H\left(  \rho_{x}\right)
\label{eq:CQE-conv-Alice-sec-bound-1}\\
&  \leq H\left(  \rho\right) \\
&  \leq ng\left(  N_{S}\right)  , \label{eq:CQE-conv-Alice-sec-bound-3}%
\end{align}
where the second inequality follows from concavity of quantum entropy, and the
last follows from the fact that the maximum entropy for a bosonic state of
mean photon number$~N$ is $g\left(  N\right)  $. Thus, there exists some
$\lambda^{\prime}\in\left[  0,1\right]  $ such that%
\begin{equation}
\sum_{x}p_{X}\left(  x\right)  H\left(  \rho_{x}\right)  =ng\left(
\lambda^{\prime}N_{S}\right)  , \label{eq:CQE-conv-Alice-ent}%
\end{equation}
because $g\left(  x\right)  $ is a monotonically increasing function of $x$
for $x\geq0$. We then have that%
\begin{equation}
\sum_{x}p_{X}\left(  x\right)  g\left(  \lambda_{x}^{\prime}N_{S,x}\right)
=g\left(  \lambda^{\prime}N_{S}\right)  , \label{eq:CQE-relate-Alice-ents}%
\end{equation}
by combining (\ref{eq:CQE-Alice-ind-ent}) and (\ref{eq:CQE-conv-Alice-ent}).
Alice can simulate Bob's state by passing her system through one input port of
a beamsplitter with transmissivity $\eta$ while passing the vacuum through the
other port. Assuming the truth of Strong Conjecture~2 from
Refs.~\cite{GSE07,G08}, we have that%
\[
H\left(  \mathcal{N}^{\otimes n}\left(  \rho_{x}\right)  \right)  \geq
ng\left(  \lambda_{x}^{\prime}\eta N_{S,x}\right)  .
\]
(We should point out that Strong Conjecture~2 holds if the entropy
photon-number inequality is true \cite{G08}). Using the relation in
(\ref{eq:CQE-relate-Alice-ents}) and concavity of $g\left(  x\right)  $, we
can apply a slightly modified version of Corollary~A.4 from Guha's thesis
\cite{G08} to show that%
\[
\sum_{x}p_{X}\left(  x\right)  g\left(  \lambda_{x}^{\prime}\eta
N_{S,x}\right)  \geq g\left(  \lambda^{\prime}\eta N_{S}\right)  ,
\]
giving the lower bound%
\begin{equation}
\sum_{x}p_{X}\left(  x\right)  H\left(  \mathcal{N}^{\otimes n}\left(
\rho_{x}\right)  \right)  \geq g\left(  \lambda^{\prime}\eta N_{S}\right)  .
\label{eq:CQE-Bob-lower-bnd}%
\end{equation}
(Corollary~A.4 of Ref.~\cite{G08} is stated for a uniform distribution, but
the argument only relies on a concavity argument and thus applies to an
arbitrary distribution.)

With a similar development as in (\ref{eq:CQE-conv-Alice-bound-1}%
-\ref{eq:CQE-conv-Alice-bound-4}), we can bound the entropy $H\left(
\mathcal{N}^{\otimes n}\left(  \rho_{x}\right)  \right)  $ because the mean
number of photons for the $j^{\text{th}}$ symbol of $\mathcal{N}^{\otimes
n}\left(  \rho_{x}\right)  $ is $\eta N_{S,x_{j}}$:%
\begin{align*}
0  &  \leq H\left(  \mathcal{N}^{\otimes n}\left(  \rho_{x}\right)  \right) \\
&  \leq\sum_{j=1}^{n}H\left(  \mathcal{N}^{\otimes n}\left(  \rho_{x}%
^{j}\right)  \right) \\
&  \leq\sum_{j=1}^{n}g\left(  \eta N_{S,x_{j}}\right) \\
&  \leq ng\left(  \eta N_{S,x}\right)  .
\end{align*}
Thus, for all $x\in\mathcal{X}$ there exists some $\lambda_{x}\in\left[
0,1\right]  $ such that%
\begin{equation}
H\left(  \mathcal{N}^{\otimes n}\left(  \rho_{x}\right)  \right)  =ng\left(
\lambda_{x}\eta N_{S,x}\right)  , \label{eq:CQE-conv-Bob-ind-bounds}%
\end{equation}
because $g\left(  x\right)  $ is a monotonically increasing function of $x$
for $x\geq0$. We also have that%
\begin{align*}
0  &  \leq\sum_{x}p_{X}\left(  x\right)  H\left(  \mathcal{N}^{\otimes
n}\left(  \rho_{x}\right)  \right) \\
&  \leq H\left(  \mathcal{N}^{\otimes n}\left(  \rho\right)  \right) \\
&  \leq ng\left(  \eta N_{S}\right)  ,
\end{align*}
for reasons similar to those in (\ref{eq:CQE-conv-Alice-sec-bound-1}%
-\ref{eq:CQE-conv-Alice-sec-bound-3}). Thus, there exists some $\lambda
\in\left[  0,1\right]  $ such that%
\begin{equation}
\sum_{x}p_{X}\left(  x\right)  H\left(  \mathcal{N}^{\otimes n}\left(
\rho_{x}\right)  \right)  =ng\left(  \lambda\eta N_{S}\right)  ,
\label{eq:CQE-Bob-entropy-1}%
\end{equation}
because $g\left(  x\right)  $ is a monotonically increasing function of $x$
for $x\geq0$. This gives us our third bound in (\ref{eq:CQE-conv-3}).
Combining (\ref{eq:CQE-Bob-entropy-1}) and (\ref{eq:CQE-Bob-lower-bnd}) gives
us the following bound%
\[
g\left(  \lambda\eta N_{S}\right)  \geq g\left(  \lambda^{\prime}\eta
N_{S}\right)  ,
\]
which in turn implies that%
\[
g\left(  \lambda N_{S}\right)  \geq g\left(  \lambda^{\prime}N_{S}\right)  ,
\]
because $g\left(  x\right)  $ and its inverse $g^{-1}\left(  y\right)  $
(defined on positive reals) are both monotonically increasing. This gives us
our first bound in (\ref{eq:CQE-conv-1}) by combining with
(\ref{eq:CQE-conv-Alice-ent}).

Combining (\ref{eq:CQE-conv-Bob-ind-bounds}) and (\ref{eq:CQE-Bob-entropy-1}),
we have%
\[
\sum_{x}p_{X}\left(  x\right)  ng\left(  \lambda_{x}N_{S,x}\right)  =ng\left(
\eta\lambda N_{S}\right)  .
\]
We are assuming that the lossy bosonic channel has $\eta\geq1/2$ so that it is
degradable and Bob can simulate Eve's state by passing his state through one
input port of a beamsplitter with transmissivity $\left(  1-\eta\right)
/\eta$ while passing the vacuum through the other port. Assuming the truth of
Strong Conjecture~2 from Refs.~\cite{GSE07,G08}, we have that%
\[
H\left(  \left(  \mathcal{N}^{c}\right)  ^{\otimes n}\left(  \rho_{x}\right)
\right)  \geq ng\left(  \lambda_{x}\left(  1-\eta\right)  N_{S,x}\right)  .
\]
Using the above relation, concavity of $g\left(  x\right)  $, and the fact
that $\eta\geq1/2$, we can apply the modified version of Corollary~A.4 from
Ref.~\cite{G08} to show that%
\[
\sum_{x}p_{X}\left(  x\right)  g\left(  \lambda_{x}\left(  1-\eta\right)
N_{S,x}\right)  \geq g\left(  \lambda\left(  1-\eta\right)  N_{S}\right)  .
\]
This gives our final bound in (\ref{eq:CQE-conv-4}):%
\begin{align*}
&  \sum_{x}p_{X}\left(  x\right)  H\left(  \left(  \mathcal{N}^{c}\right)
^{\otimes n}\left(  \rho_{x}\right)  \right) \\
&  \geq\sum_{x}p_{X}\left(  x\right)  g\left(  \lambda_{x}\left(
1-\eta\right)  N_{S,x}\right) \\
&  \geq g\left(  \lambda\left(  1-\eta\right)  N_{S}\right)  .
\end{align*}

\end{proof}

\subsubsection{Special Cases of the Quantum Dynamic Region}

A similar proof can be used to show that our characterization of the CE
trade-off curve is optimal for all $\eta\in\left[  0,1\right]  $ and that our
characterization of the CQ\ trade-off curve is optimal for all $\eta\in\left[
1/2,1\right]  $. Of course, both proofs require Strong Conjecture~2 from
Refs.~\cite{GSE07,G08} (which holds if the entropy photon-number inequality is true).

We also can completely characterize the trade-off between quantum
communication and entanglement consumption for a lossy bosonic channel with
$\eta\geq1/2$.

\begin{theorem}
The trade-off between entanglement assistance and quantum communication (when
$C=0$, $Q\geq0$, and $E\leq0$) given by (\ref{eq:bosonic-region-1}%
-\ref{eq:bosonic-region-3}) is optimal for the lossy bosonic channel with
$\eta\geq1/2$.
\end{theorem}

\begin{proof}
Recall from Ref.~\cite{DHW05RI}\ that the characterization of the
entanglement-assisted quantum capacity region for any channel $\mathcal{N}$ is
the regularization the union of the following regions:%
\begin{align*}
2Q  &  \leq I\left(  A;B\right)  _{\rho},\\
Q  &  \leq I\left(  A\rangle B\right)  _{\rho}+\left\vert E\right\vert ,
\end{align*}
where the entropies are with respect to a state $\rho^{AB}\equiv
\mathcal{N}^{A^{\prime}\rightarrow B}(\phi^{AA^{\prime}})$, the union of the
above regions is over pure, bipartite states $\phi^{AA^{\prime}}$, $Q$ is the
rate of quantum communication, and $E$ is the rate of entanglement
consumption. Characterizing the boundary of the region is equivalent to
optimizing the following function~\cite{DHW05RI}:%
\[
\max_{\rho}I\left(  A;B\right)  _{\rho}+\mu I\left(  A\rangle B\right)
_{\rho},
\]
where $\mu$ is a positive number playing the role of a Lagrange multiplier. We
can rewrite the above function in the following form:%
\[
\max_{\rho}H\left(  A\right)  _{\rho}+\left(  \mu+1\right)  I\left(  A\rangle
B\right)  _{\rho},
\]
because $I\left(  A;B\right)  =H\left(  A\right)  +I\left(  A\rangle B\right)
$. It is straightforward to show that the above formula is additive for
degradable channels~\cite{DS03}, and furthermore, we can rewrite the coherent
information $I\left(  A\rangle B\right)  _{\rho}$ as a conditional entropy
$H\left(  F|E\right)  $, whenever the channel is degradable~\cite{DS03} (let
$F$ be the environment of the degrading map). Then, from the extremality of
Gaussian states for entropy and conditional entropy~\cite{EW07,WGC06}, it
suffices to perform the above optimization over only Gaussian states. Recall
that $I\left(  A\rangle B\right)  =H\left(  B\right)  -H\left(  E\right)  $
where $E$ is the environment of the channel. For a lossy bosonic channel with
$\eta\geq1/2$ with mean input photon number~$N_{S}$, the following bounds hold%
\begin{align*}
0  &  \leq H\left(  A\right)  _{\rho}\leq g\left(  N_{S}\right)  ,\\
0  &  \leq H\left(  B\right)  _{\rho}\leq g\left(  \eta N_{S}\right)  .
\end{align*}
Thus, there exist some $\lambda^{\prime},\lambda\in\left[  0,1\right]  $ such
that $H\left(  A\right)  _{\rho}=g\left(  \lambda^{\prime}N_{S}\right)  $ and
$H\left(  B\right)  _{\rho}=g\left(  \eta\lambda N_{S}\right)  $ from the
monotonicity of $g\left(  x\right)  $. Also, we have that $g\left(
\eta\lambda N_{S}\right)  \geq g\left(  \eta\lambda^{\prime}N_{S}\right)  $
from the same reasoning as in the above proof, though Strong Conjecture~2 from
Ref.~\cite{GSE07,G08}\ is known to hold for Gaussian states. This then implies
that $g\left(  \lambda N_{S}\right)  \geq g\left(  \lambda^{\prime}%
N_{S}\right)  $ by the same reasoning as in the above proof. Also, we have the
following lower bound from Strong Conjecture~2 (which holds for Gaussian
states):%
\[
H\left(  E\right)  _{\rho}\geq g\left(  \left(  1-\eta\right)  \lambda
N_{S}\right)  ,
\]
by the same reasoning as in the above proof. Thus, we have the following upper
bound for a particular state $\rho$:%
\begin{multline*}
H\left(  A\right)  _{\rho}+\left(  \mu+1\right)  I\left(  A\rangle B\right)
_{\rho}\\
\leq g\left(  \eta\lambda N_{S}\right)  +\left(  \mu+1\right)  \left[
g\left(  \eta\lambda N_{S}\right)  -g\left(  \left(  1-\eta\right)  \lambda
N_{S}\right)  \right]  .
\end{multline*}
It is straightforward to show that $g\left(  \eta x\right)  -g\left(  \left(
1-\eta\right)  x\right)  $ is a monotonically increasing function in $x$
whenever $\eta\geq1/2$, by considering that $g\left(  \eta x\right)  =g\left(
\left(  1-\eta\right)  x\right)  =0$ for $x=0$, $g\left(  \eta x\right)  \geq
g\left(  \left(  1-\eta\right)  x\right)  $, and $\frac{\partial}{\partial
x}g\left(  \eta x\right)  \geq\frac{\partial}{\partial x}g\left(  \left(
1-\eta\right)  x\right)  $ whenever $\eta\geq1/2$. Thus, we obtain our final
upper bound by setting $\lambda=1$. These rates are achievable simply taking
the input state $\phi^{A^{\prime}}$ to be thermal with mean photon number
$N_{S}$.
\end{proof}

\subsection{Converse for the Private Dynamic Region}

We can prove the converse for the private dynamic capacity region similarly to
how we did for the quantum dynamic capacity region.
Proposition~\ref{prop:RPS-dynamic} states that the regularization of the
region in (\ref{eq:private-dynamic-region-1}-\ref{eq:private-dynamic-region-3}%
) is equivalent to the private dynamic capacity region. We prove that the
region in (\ref{eq:private-dynamic-1}-\ref{eq:private-dynamic-3}) is
equivalent to the capacity region for a lossy bosonic channel with
transmissivity parameter $\eta>1/2$ and mean input photon number$~N_{S}$, up
to a minimum output entropy conjecture. We do so by proving the following
upper bounds:%
\begin{align}
H\left(  \mathcal{N}^{\otimes n}\left(  \rho\right)  \right)   &  \leq
ng\left(  \eta N_{S}\right)  ,\\
\sum_{x}p_{X}\left(  x\right)  H\left(  \mathcal{N}^{\otimes n}\left(
\rho_{x}\right)  \right)   &  \leq ng\left(  \eta\lambda N_{S}\right)  ,
\end{align}
and the following lower bounds (the first up to the minimum output entropy
conjecture):%
\begin{align}
\sum_{x}p_{X}\left(  x\right)  H\left(  \left(  \mathcal{N}^{c}\right)
^{\otimes n}\left(  \rho_{x}\right)  \right)   &  \geq ng\left(  \left(
1-\eta\right)  \lambda N_{S}\right)  ,\nonumber\\
\sum_{x,y}p\left(  x\right)  p\left(  y|x\right)  H\left(  \mathcal{N}%
^{\otimes n}\left(  \rho_{x,y}\right)  \right)   &  \geq0.
\label{eq:RPS-conv-4}%
\end{align}
so that for all $n$-letter ensembles $\left\{  p_{X}\left(  x\right)
p_{Y|X}\left(  y|x\right)  ,\rho_{x,y}\right\}  $ with $\rho_{x,y}%
\in\mathcal{B}\left(  \mathcal{H}^{\otimes n}\right)  $, $\rho_{x}\equiv
\sum_{y}p_{Y|X}\left(  y|x\right)  \rho_{x,y}$, and $\rho\equiv\sum_{x}%
p_{X}\left(  x\right)  \rho_{x}$, there exists some $\lambda\in\left[
0,1\right]  $ such that the above bounds hold. The above bounds immediately
imply that the region in (\ref{eq:private-dynamic-1}%
-\ref{eq:private-dynamic-3}) is the private dynamic capacity region of the
lossy bosonic channel.

\begin{proof}
The last bound in (\ref{eq:RPS-conv-4}) follows simply because the quantum
entropy is always positive. The other bounds follow by the same method given
in the converse of the quantum dynamic capacity region for the lossy bosonic channel.
\end{proof}

\section{The Thermal Noise Channel}

\label{sec:thermal}The thermal noise channel is the same map as in
(\ref{eq:lossy-bosonic-channel-1}-\ref{eq:lossy-bosonic-channel-2}), with the
exception that the environment is in a thermal state with mean photon
number$~N_{B}$.

\subsection{Quantum Dynamic Region}

\begin{theorem}
An achievable quantum dynamic region for the thermal noise channel with
transmissivity $\eta$ and mean thermal photon number $N_{B}$ is as follows:%
\begin{align*}
C+2Q  &  \leq g\left(  \lambda N_{S}\right)  +g\left(  \eta N_{S}+\left(
1-\eta\right)  N_{B}\right) \\
&  \ \ \ \ \ \ \ \ -g_{E}^{\eta}\left(  \lambda,N_{S},N_{B}\right)  ,\\
Q+E  &  \leq g\left(  \eta\lambda N_{S}+\left(  1-\eta\right)  N_{B}\right)
-g_{E}^{\eta}\left(  \lambda,N_{S},N_{B}\right)  ,\\
C+Q+E  &  \leq g\left(  \eta N_{S}+\left(  1-\eta\right)  N_{B}\right)
-g_{E}^{\eta}\left(  \lambda,N_{S},N_{B}\right)  ,
\end{align*}
where $\lambda\in\left[  0,1\right]  $ is a photon-number-sharing parameter,
$g\left(  N\right)  $ is defined in (2), and%
\begin{multline*}
g_{E}^{\eta}\left(  \lambda,N_{S},N_{B}\right)  \equiv\\
g\left(  \left[  D+\lambda\left(  1-\eta\right)  N_{S}-\left(  1-\eta\right)
N_{B}-1\right]  /2\right) \\
+g\left(  \left[  D-\lambda\left(  1-\eta\right)  N_{S}+\left(  1-\eta\right)
N_{B}-1\right]  /2\right)  ,
\end{multline*}%
\begin{multline*}
D^{2}\equiv\left[  \lambda\left(  1+\eta\right)  N_{S}+\left(  1-\eta\right)
N_{B}+1\right]  ^{2}\\
-4\eta\lambda N_{S}\left(  \lambda N_{S}+1\right)  .
\end{multline*}

\end{theorem}

\begin{proof}
We use the same coding strategy as in (\ref{eq:bosonic-ensemble}) and then
simply need to calculate the four entropies in
(\ref{eq:bosonic-CQE-entropies-1}-\ref{eq:bosonic-CQE-entropies-4}) for this
case. The average output state is a thermal state with mean number of photons
$\eta N_{S}+\left(  1-\eta\right)  N_{B}$, implying that the first entropy in
(\ref{eq:bosonic-CQE-entropies-1}) is%
\[
H\left(  \mathcal{N}\left(  \overline{\theta}\right)  \right)  =g\left(  \eta
N_{S}+\left(  1-\eta\right)  N_{B}\right)  .
\]
The second entropy in (\ref{eq:bosonic-CQE-entropies-2}) is the same as before
because it is the entropy of the half of the state not transmitted through the
channel:%
\[
\int d\alpha\ p_{\overline{\lambda}N_{S}}\left(  \alpha\right)  \ H\left(
D\left(  \alpha\right)  \theta D^{\dag}\left(  \alpha\right)  \right)
=g\left(  \lambda N_{S}\right)  .
\]
The state of the output conditioned on the displacement operator applied is a
thermal state with mean photon number $\eta\lambda N_{S}+\left(
1-\eta\right)  N_{B}$. Thus, the third entropy in
(\ref{eq:bosonic-CQE-entropies-3}) is%
\begin{multline*}
\int d\alpha\ p_{\overline{\lambda}N_{S}}\left(  \alpha\right)  \ H\left(
\mathcal{N}\left(  D\left(  \alpha\right)  \theta D^{\dag}\left(
\alpha\right)  \right)  \right) \\
=g\left(  \eta\lambda N_{S}+\left(  1-\eta\right)  N_{B}\right)  .
\end{multline*}
We calculate the fourth entropy in (\ref{eq:bosonic-CQE-entropies-4}) in a
different way, along the lines presented in Refs.~\cite{HW01,GLMS03a}. The
displacement operator does not affect the correlation matrix of the two-mode
squeezed state, and thus the entropy of this state does not change under such
a transformation. In this case, the entropy is%
\begin{multline}
g_{E}^{\eta}\left(  \lambda,N_{S},N_{B}\right)  \equiv
\label{eq:entropy-eve-thermal}\\
g\left(  \left[  D+\lambda\left(  1-\eta\right)  N_{S}-\left(  1-\eta\right)
N_{B}-1\right]  /2\right) \\
+g\left(  \left[  D-\lambda\left(  1-\eta\right)  N_{S}+\left(  1-\eta\right)
N_{B}-1\right]  /2\right)  ,
\end{multline}
where%
\begin{multline*}
D^{2}\equiv\left[  \lambda\left(  1+\eta\right)  N_{S}+\left(  1-\eta\right)
N_{B}+1\right]  ^{2}\\
-4\eta\lambda N_{S}\left(  \lambda N_{S}+1\right)  .
\end{multline*}
Plugging in the entropies gives us the statement of the proposition.
\end{proof}

The trade-off region for classical and quantum communication is%
\begin{align}
Q  &  \leq g\left(  \eta\left(  \lambda\right)  N_{S}+\left(  1-\eta\right)
N_{B}\right) \nonumber\\
&  \ \ \ \ \ \ \ \ -g_{E}^{\eta}\left(  \lambda,N_{S},N_{B}\right)
,\label{eq:CQ-thermal-1}\\
C+Q  &  \leq g\left(  \eta N_{S}+\left(  1-\eta\right)  N_{B}\right)
-g_{E}^{\eta}\left(  \lambda,N_{S},N_{B}\right)  . \label{eq:CQ-thermal-2}%
\end{align}
The trade-off region for assisted and unassisted classical communication is%
\begin{align*}
C  &  \leq g\left(  \left(  \lambda\right)  N_{S}\right)  +g\left(  \eta
N_{S}+\left(  1-\eta\right)  N_{B}\right) \\
&  \ \ \ \ \ \ \ \ \ \ -g_{E}^{\eta}\left(  \lambda,N_{S},N_{B}\right)  ,\\
C  &  \leq g\left(  \eta N_{S}+\left(  1-\eta\right)  N_{B}\right)
-g_{E}^{\eta}\left(  \lambda,N_{S},N_{B}\right)  +E,
\end{align*}
where above our convention changes so that positive~$E$ corresponds to the
consumption of entanglement. Figure~\ref{fig:thermalizing-bosonic}(a) depicts
the trade-off between classical and quantum communication for a thermal noise
channel, and Figure~\ref{fig:thermalizing-bosonic}(b) depicts the trade-off
between assisted and unassisted classical communication.%
\begin{figure*}
[ptb]
\begin{center}
\includegraphics[
natheight=4.159700in,
natwidth=9.280300in,
width=6.7525in
]%
{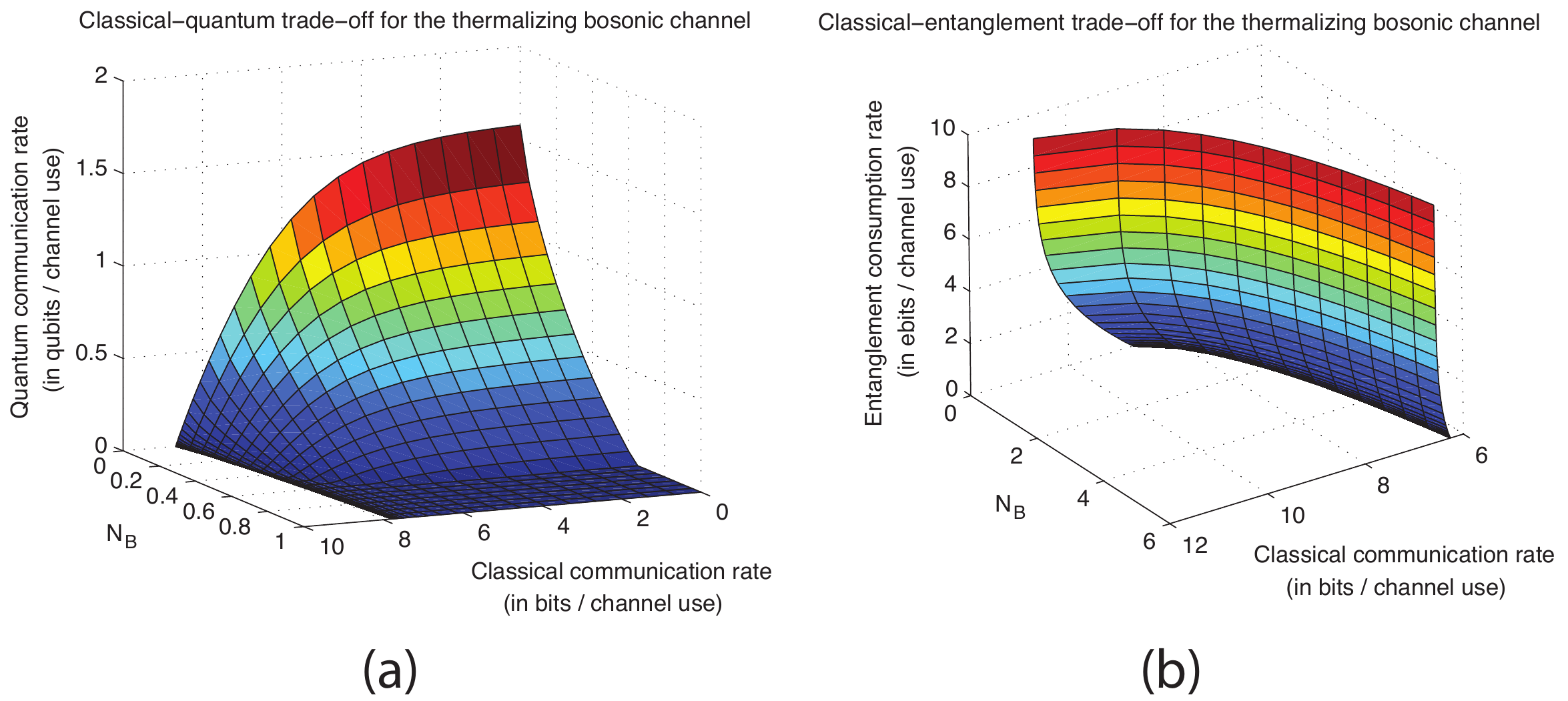}%
\caption{(Color online) This figure displays the effect of increasing thermal noise
(parametrized by~$N_{B}$) on the trade-off between (a)\ classical and quantum
communication and (b) entanglement-assisted and unassisted classical
communication. We cannot say whether any points in these regions are optimal
because the capacity of the thermal channel is unknown (though, they are known
if the minimum-output entropy conjecture is true~\cite{GGLMS04,GHLM10}).}%
\label{fig:thermalizing-bosonic}%
\end{center}
\end{figure*}

\subsection{Private Dynamic Region}

\begin{theorem}
An achievable private dynamic region for the thermal noise channel with
transmissivity $\eta$ and mean thermal photon number $N_{B}$ is as follows:%
\begin{align*}
R+P  &  \leq g\left(  \eta N_{S}+\left(  1-\eta\right)  N_{B}\right)
-g\left(  \left(  1-\eta\right)  N_{B}\right)  ,\\
P+S  &  \leq g\left(  \eta\lambda N_{S}+\left(  1-\eta\right)  N_{B}\right)
-g_{E}^{\eta}\left(  \lambda,N_{S},N_{B}\right)  ,\\
R+P+S  &  \leq g\left(  \eta N_{S}+\left(  1-\eta\right)  N_{B}\right)
-g_{E}^{\eta}\left(  \lambda,N_{S},N_{B}\right)  ,
\end{align*}
where $\lambda\in\left[  0,1\right]  $ is a photon-number-sharing parameter,
$g\left(  N\right)  $ is defined in (2), and $g_{E}^{\eta}$ is defined in
(\ref{eq:entropy-eve-thermal}).
\end{theorem}

\begin{proof}
We can obtain an expression for a private dynamic achievable rate region of
the thermal noise channel. If we use the same ensemble for coding as in
(\ref{eq:ensemble-public-private}), then the resulting private dynamic
achievable rate region is slightly different from that in
(\ref{eq:private-dynamic-1}-\ref{eq:private-dynamic-3}). A thermal channel
does not preserve the purity of coherent states transmitted through it. Thus,
the entropy in (\ref{eq:private-conditional-coherent-entropy}) is no longer
equal to zero, but it is instead equal to%
\begin{multline*}
\int\int d\alpha\ d\beta\ p_{\overline{\lambda}N_{S}}\left(  \alpha\right)
p_{\left(  \lambda\right)  N_{S}}\left(  \beta\right)  \ H\left(
\mathcal{N}\left(  \left\vert \alpha+\beta\right\rangle \left\langle
\alpha+\beta\right\vert \right)  \right) \\
=g\left(  \left(  1-\eta\right)  N_{B}\right)  ,
\end{multline*}
because Bob's state is a displaced thermal state with mean photon number
$\left(  1-\eta\right)  N_{B}$ (the amount of noise that the environment
injects into the state). Thus, the expression for the private dynamic
achievable rate region is as stated in the theorem.
\end{proof}

This region is generally smaller than the full quantum dynamic achievable rate
region for the same values of $N_{S}$, $N_{B}$, and $\eta$ because there is no
analog of the superdense coding effect with the resources of public classical
communication, private classical communication, and secret key
\cite{WH10a,CP02}. Though, the trade-off between public and private
communication with this coding strategy is the same as that between classical
and quantum communication in (\ref{eq:CQ-thermal-1}-\ref{eq:CQ-thermal-2}).

\section{The Amplifying Channel}

\label{sec:amplifier}The amplifying channel is another bosonic channel
important in applications, modeling any kind of amplification process that can
occur in bosonic systems. These applications range from cavities coupled with
Josephson junctions~\cite{CDGMS10}, to non-degenerate parametric amplifiers in
quantum optics~\cite{C82}, to the Unruh effect from relativistic quantum
mechanics~\cite{PhysRevD.14.870}.
%Br\'{a}dler \textit{et al}.~have studied the
%consequences of the Unruh effect for capacities of quantum channels in some
%detail, under the assumption that the receiver is uniformly accelerating with
%respect to the sender and the sender encodes information into a single
%excitation of multiple
%modes~\cite{BHP09,Bradler_Cloning,bradler:072201,BHTW10,Bradler_Rindler,JBW11}%
%. Our results in this section apply to the more general setting where we
%impose only a finite photon-number constraint on the input modes.

The amplifying channel corresponds to the following transformation of the
input annihilation operator:%
\begin{align}
\hat{a}  &  \rightarrow\sqrt{\kappa}\ \hat{a}+\sqrt{\kappa-1}\ \hat{e}^{\dag
},\label{eq:amplify-channel-1}\\
\hat{e}^{\dag}  &  \rightarrow\sqrt{\kappa-1}\ \hat{a}+\sqrt{\kappa}\ \hat
{e}^{\dag},\label{eq:amplify-channel-2}%
\end{align}
where $\kappa\geq1$ is the amplifier gain and $\hat{e}$ is now an auxiliary
mode associated with the amplification process. If the state of the auxiliary
mode is a vacuum state, then this auxiliary mode injects the minimum possible
noise into the signal mode. We can consider the auxiliary mode more generally
to be in a thermal state with mean photon number$~N_{B}$.

\subsection{Achievable Quantum and Private Dynamic Regions}

\begin{theorem}
\label{thm:CQE-amplify}An achievable quantum dynamic region for the amplifying
channel with gain $\kappa\geq1$ and mean thermal photon number $N_{B}$ is as
follows:%
\begin{align*}
C+2Q  &  \leq g\left(  \lambda N_{S}\right)  +g\left(  \kappa N_{S}+\left(
\kappa-1\right)  \left(  N_{B}+1\right)  \right) \\
&  \ \ \ \ \ -g_{E}^{\kappa}\left(  \lambda,N_{S},N_{B}\right)  ,\\
Q+E  &  \leq g\left(  \kappa\lambda N_{S}+\left(  \kappa-1\right)  \left(
N_{B}+1\right)  \right) \\
&  \ \ \ \ \ -g_{E}^{\kappa}\left(  \lambda,N_{S},N_{B}\right)  ,\\
C+Q+E  &  \leq g\left(  \kappa N_{S}+\left(  \kappa-1\right)  \left(
N_{B}+1\right)  \right) \\
&  \ \ \ \ \ -g_{E}^{\kappa}\left(  \lambda,N_{S},N_{B}\right)  ,
\end{align*}
where $\lambda\in\left[  0,1\right]  $ is a photon-number-sharing parameter,
$g\left(  N\right)  $ is defined in (2), and%
\begin{multline*}
g_{E}^{\kappa}\left(  \lambda,N_{S},N_{B}\right)  \equiv\\
g\left(  \left[  D+\left(  \kappa-1\right)  \left[  \lambda N_{S}%
+N_{B}+1\right]  -1\right]  /2\right) \\
+g\left(  \left[  D-\left(  \kappa-1\right)  \left[  \lambda N_{S}%
+N_{B}+1\right]  -1\right]  /2\right)  ,
\end{multline*}%
\begin{multline*}
D^{2}\equiv\left[  \lambda\left(  1+\kappa\right)  N_{S}+\left(
\kappa-1\right)  \left(  N_{B}+1\right)  +1\right]  ^{2}\\
-4\kappa\lambda N_{S}\left(  \lambda N_{S}+1\right)  .
\end{multline*}

\end{theorem}

\begin{theorem}
\label{thm:RPS-amplify}An achievable private dynamic region for the amplifying
channel with gain $\kappa\geq1$ and mean thermal photon number $N_{B}$ is as
follows:%
\begin{align*}
R+P  &  \leq g\left(  \kappa N_{S}+\left(  \kappa-1\right)  \left(
N_{B}+1\right)  \right) \\
&  \ \ \ \ \ \ \ \ -g\left(  \left(  \kappa-1\right)  \left(  N_{B}+1\right)
\right)  ,\\
P+S  &  \leq g\left(  \kappa\lambda N_{S}+\left(  \kappa-1\right)  \left(
N_{B}+1\right)  \right) \\
&  \ \ \ \ \ \ \ \ -g_{E}^{\kappa}\left(  \lambda,N_{S},N_{B}\right)  ,\\
R+P+S  &  \leq g\left(  \kappa N_{S}+\left(  \kappa-1\right)  \left(
N_{B}+1\right)  \right) \\
&  \ \ \ \ \ \ \ \ -g_{E}^{\kappa}\left(  \lambda,N_{S},N_{B}\right)  ,
\end{align*}
where $\lambda\in\left[  0,1\right]  $ is a photon-number-sharing parameter,
$g\left(  N\right)  $ is defined in (2), and $g_{E}^{\kappa}$ is defined in
(\ref{eq:eve-entropy-amplify}).
\end{theorem}

\begin{proof}
\textbf{(Of Theorems \ref{thm:CQE-amplify} and \ref{thm:RPS-amplify})} The
trade-off coding scheme for both the quantum dynamic and private dynamic
trade-off settings is the same as we used before in (\ref{eq:bosonic-ensemble}%
) and (\ref{eq:ensemble-public-private}), respectively. Thus, we only need to
calculate the various entropies associated with the amplifying channel. We
consider the quantum dynamic setting first and calculate the four entropies in
(\ref{eq:bosonic-CQE-entropies-1}-\ref{eq:bosonic-CQE-entropies-4}).

The state resulting from transmitting a thermal state with mean photon number
$N_{S}$ through the amplifying channel is a thermal state with mean photon
number $\kappa N_{S}+\left(  \kappa-1\right)  \left(  N_{B}+1\right)  $. Thus,
it follows that%
\[
H\left(  \mathcal{N}\left(  \overline{\theta}\right)  \right)  =g\left(
\kappa N_{S}+\left(  \kappa-1\right)  \left(  N_{B}+1\right)  \right)  .
\]
By the same argument as before, we have that%
\[
\int d\alpha\ p_{\overline{\lambda}N_{S}}\left(  \alpha\right)  \ H\left(
D\left(  \alpha\right)  \theta D^{\dag}\left(  \alpha\right)  \right)
=g\left(  \lambda N_{S}\right)  .
\]
The displacement operators acting on a thermal state are again covariant with
respect to the amplifying channel so that%
\begin{multline*}
\int d\alpha\ p_{\overline{\lambda}N_{S}}\left(  \alpha\right)  \ H\left(
\mathcal{N}\left(  D\left(  \alpha\right)  \theta D^{\dag}\left(
\alpha\right)  \right)  \right) \\
=g\left(  \kappa\lambda N_{S}+\left(  \kappa-1\right)  \left(  N_{B}+1\right)
\right)  .
\end{multline*}
Finally, we can make use of the Holevo-Werner results in Section~V-A of
Ref.~\cite{HW01} to show that%
\begin{multline*}
\int d\alpha\ p_{\overline{\lambda}N_{S}}\left(  \alpha\right)  \ H\left(
\mathcal{N}^{c}\left(  D\left(  \alpha\right)  \theta D^{\dag}\left(
\alpha\right)  \right)  \right) \\
=g_{E}^{\kappa}\left(  \lambda,N_{S},N_{B}\right)  ,
\end{multline*}
where%
\begin{multline}
g_{E}^{\kappa}\left(  \lambda,N_{S},N_{B}\right)  \equiv
\label{eq:eve-entropy-amplify}\\
g\left(  \left[  D+\left(  \kappa-1\right)  \left[  \lambda N_{S}%
+N_{B}+1\right]  -1\right]  /2\right) \\
+g\left(  \left[  D-\left(  \kappa-1\right)  \left[  \lambda N_{S}%
+N_{B}+1\right]  -1\right]  /2\right)  ,
\end{multline}
and%
\begin{multline*}
D^{2}\equiv\left[  \lambda\left(  1+\kappa\right)  N_{S}+\left(
\kappa-1\right)  \left(  N_{B}+1\right)  +1\right]  ^{2}\\
-4\kappa\lambda N_{S}\left(  \lambda N_{S}+1\right)  .
\end{multline*}
Thus, the expression for the quantum dynamic capacity region of an amplifying
channel is as stated in Theorem~\ref{thm:CQE-amplify} above.

We can also use a similar coding scheme as in
(\ref{eq:ensemble-public-private}) for the private dynamic setting. The only
other entropy that we need to calculate for this setting is the one in
(\ref{eq:private-conditional-coherent-entropy}):%
\begin{multline*}
\int\int d\alpha\ d\beta\ p_{\overline{\lambda}N_{S}}\left(  \alpha\right)
p_{\left(  \lambda\right)  N_{S}}\left(  \beta\right)  \ H\left(
\mathcal{N}\left(  \left\vert \alpha+\beta\right\rangle \left\langle
\alpha+\beta\right\vert \right)  \right) \\
=g\left(  \left(  \kappa- 1\right)  \left( N_{B}+1\right) \right)  ,
\end{multline*}
because Bob's state is a displaced thermal state with mean photon number
$\left(  \kappa- 1\right)  \left( N_{B}+1\right) $ (the amount of noise that
the environment injects into the state). Thus, the expression for the private
dynamic achievable rate region is as stated in Theorem~\ref{thm:RPS-amplify} above.
\end{proof}

Figure~\ref{fig:amp-bosonic}\ plots the classical-quantum and
classical-entanglement trade-off curves for an amplifying channel with
$N_{S}=200$, $N_{B}=0$, and increasing values of the amplification
parameter~$\kappa$. The figures demonstrate that increased amplification
decreases performance, but the upshot is that both trade-off settings still
exhibit a remarkable improvement over time-sharing.%
\begin{figure*}
[ptb]
\begin{center}
\includegraphics[
natheight=4.159700in,
natwidth=9.033000in,
height=3.1254in,
width=6.7525in
]%
{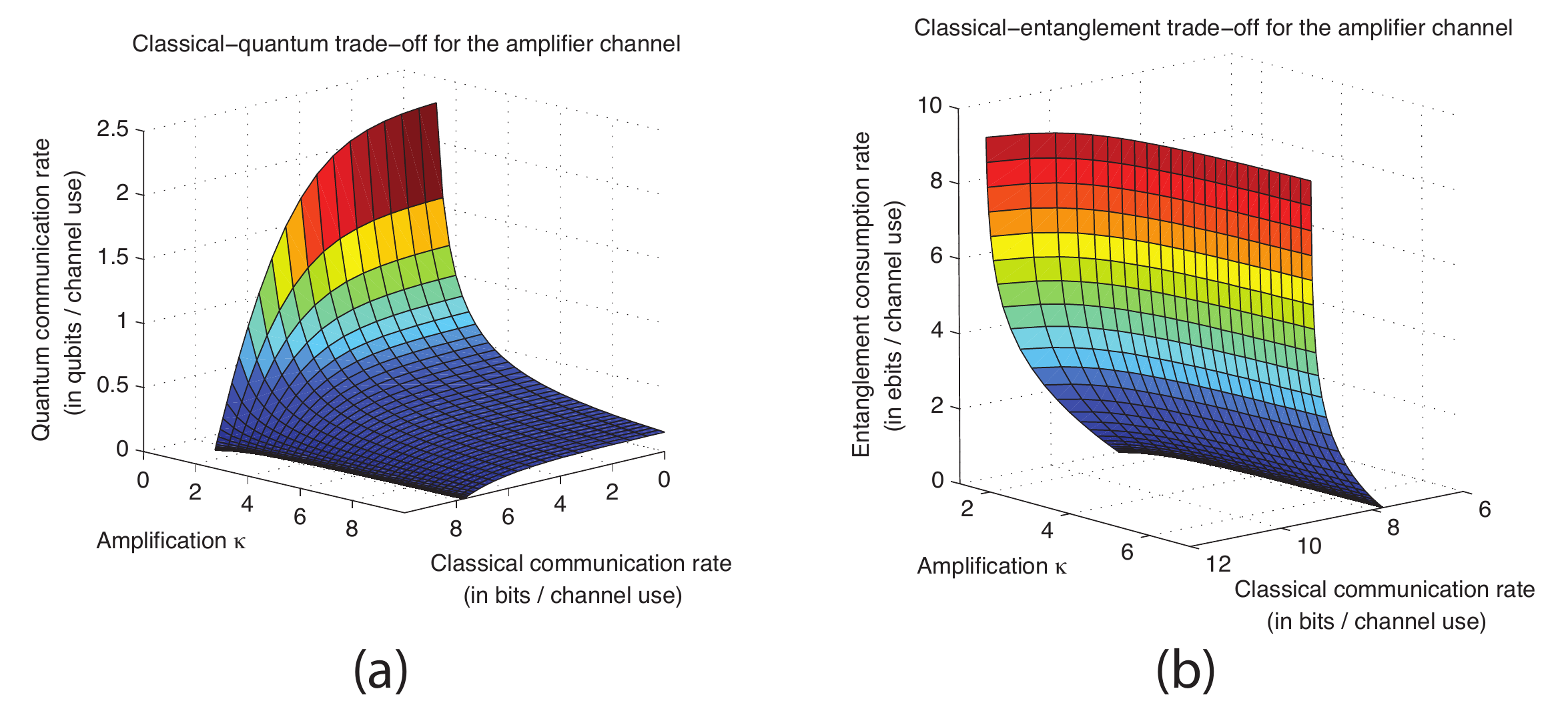}%
\caption{(Color online) (a)\ Trade-off between classical and quantum communication for the
bosonic amplifying channel with $N_{S}=200$ and $N_{B}=0$ for increasing
values of the amplification parameter~$\kappa$. (b) Trade-off between
entanglement-assisted and unassisted classical communication for the same
channel. For both trade-offs, more amplification degrades performance because
it introduces too much noise.}%
\label{fig:amp-bosonic}%
\end{center}
\end{figure*}

\section{Unruh Channel}

\label{sec:unruh}

Br\'{a}dler \textit{et al}.~studied the information theoretic consequences of
the Unruh effect in a series of
papers~\cite{BHP09,Bradler_Cloning,bradler:072201,BHTW10,Bradler_Rindler,JBW11}%
. They dubbed \textquotedblleft the Unruh channel\textquotedblright\ as the
channel induced by the Unruh effect when an observer encodes a qubit
into a single-excitation subspace of her Unruh modes and a uniformly accelerating Rindler
observer detects this state (by single-excitation encoding, we mean a
dual-rail encoding with $\left\vert 0_{L}\right\rangle \equiv\left\vert
01\right\rangle $ and $\left\vert 1_{L}\right\rangle \equiv\left\vert
10\right\rangle $). It is well known that the transformation corresponding to
the Unruh effect is equivalent to the transformation in
(\ref{eq:amplify-channel-1}-\ref{eq:amplify-channel-2}) for an amplifying
bosonic channel~\cite{NJBN11}. The results of Br\'{a}dler \textit{et
al}.~demonstrate that \textquotedblleft the Unruh channel\textquotedblright%
\ has a beautiful structure as a countably infinite direct sum of universal
cloning machine channels~\cite{BHP09}, and this property implies that both the
quantum dynamic and private dynamic capacity regions are
single-letter~\cite{BHTW10,WH10,WH10a}. More generally, their results with
single-excitation encodings of course apply to amplifying bosonic channels.%

\begin{figure*}
[ptb]
\begin{center}
\includegraphics[
natheight=4.386300in,
natwidth=11.153500in,
width=6.7525in
]%
{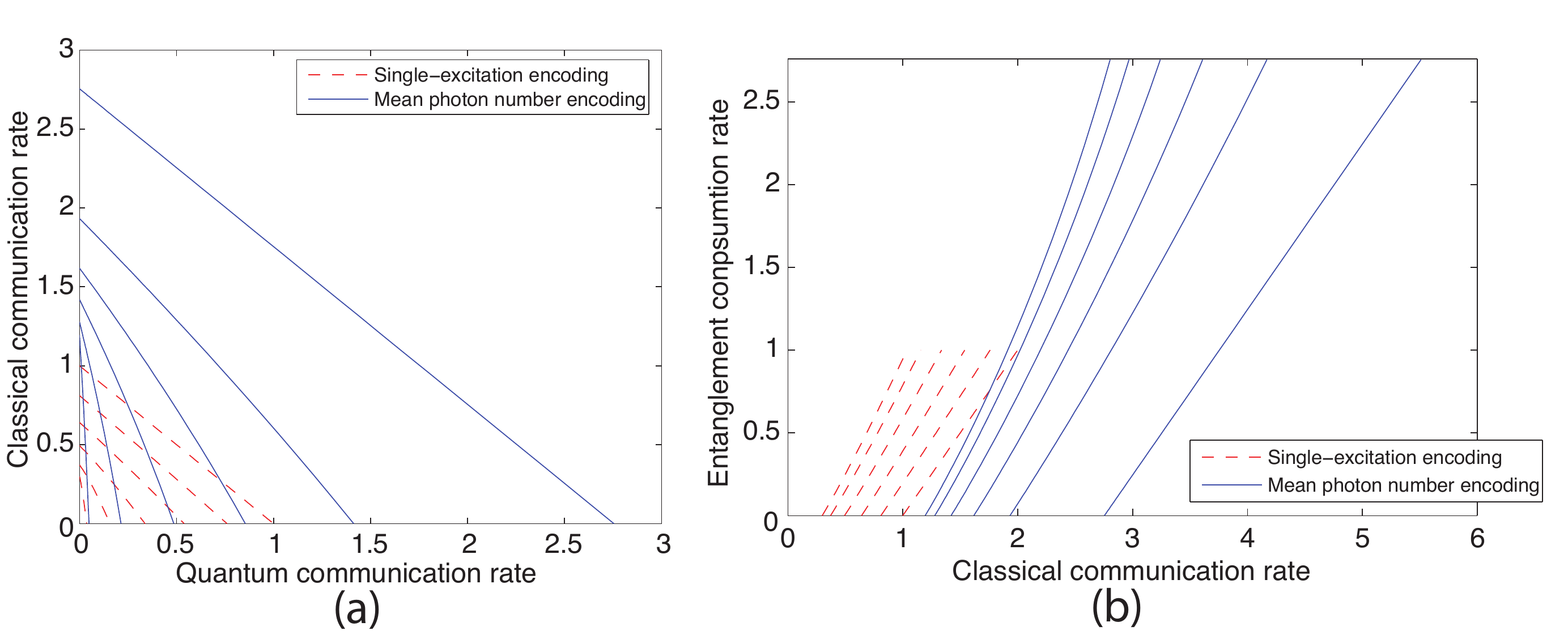}%
\caption{(Color online) Trade-off curves between (a) classical and quantum communication and
(b) entanglement-assisted and unassisted classical communication. Units for classical
communication, quantum communication, and entanglement consumption are bits per channel use,
qubits per channel use, and ebits per channel use, respectively. Each figure
plots a trade-off curve (rightmost to leftmost) for increasing values of the
acceleration parameter~$z\in\left\{  0,0.2,0.4,0.6,0.8,0.95\right\}  $. The
dotted curves are the trade-off curves from Ref.~\cite{BHTW10}\ when the
encoding is restricted to a single-excitation subspace. The solid lines are
achievable rates with a mean photon-number constraint of $1/2$. The result is
that the mean photon-number constrained encoding always outperforms
single-excitation encoding.}%
\label{fig:unruh-comparison}%
\end{center}
\end{figure*}
In spite of these analytical results, one might question calling this channel
\textit{the} Unruh channel because the encoding has a specific form as a
dual-rail encoding. More generally, we could study the capacities of the
transformation corresponding to the Unruh effect by imposing a
mean-photon-number constraint at the input, rather than restricting the form
of the encoding. In this way, we can relate the achievable rates in
Appendix~\ref{sec:amplifier} to the capacity results of Br\'{a}dler \textit{et
al}. In order to make a fair comparison, we should restrict the mean number of
photons at the input to be 1/2 because this is the mean number of photons when
sending a dual-rail maximally mixed state of the form $\left(  \left\vert
01\right\rangle \left\langle 01\right\vert +\left\vert 10\right\rangle
\left\langle 10\right\vert \right)  /2$, but we should then multiply all rates
by two because the Unruh channel of Br\'{a}dler \textit{et al}.~exploits two
uses of the Unruh transformation for one use of the Unruh channel. It is fair
to consider the maximally mixed state as input because this is the state that
achieves the boundary of the various capacity regions when tracing over all
other systems not input to the channel~\cite{BHTW10}. Furthermore, note that
the acceleration parameter~$z$ of Br\'{a}dler \textit{et al}.~in
Refs.~\cite{BHP09,BHTW10}\ is related to the amplification parameter$~\kappa$
for the amplifying bosonic channel via $z=(\kappa-1)/\kappa$.

Figure~\ref{fig:unruh-comparison}\ compares the trade-off curves for the
\textquotedblleft Unruh channel\textquotedblright\ with the trade-off curves
for the Unruh transformation with a mean photon-number constraint. The result
is that the latter outperforms the former for all depicted values of the
acceleration parameter~$z$ (though note that the quantum rates become
comparable as the acceleration parameter$~z$ increases to 0.95). The result in
the figure is unsurprising because an encoding with a mean photon-number
constraint has access to more of Fock space than does a restricted encoding. A
similar result occurs when comparing the amplitude damping channel with the
lossy bosonic channel~\cite{GF05}\ (the amplitude damping channel results from
sending a superposition of the vacuum state and a single-photon state through
the lossy bosonic channel).

\section{Conclusion}

In summary, we have provided detailed derivations of the main results
announced in our previous paper \cite{WHG12}. In particular, we have shown
that the rate regions given in (1) and (2) characterize the
classical-quantum-entanglement trade-off and the public-private-secret-key
trade-off, respectively, for communicating over a pure-loss bosonic channel.
We have argued for a \textquotedblleft rule of thumb\textquotedblright\ for
trade-off coding, so that a sender and receiver can make the best use of
photon-number sharing if a large number of photons are available on average
for coding. We have also argued that the regions in (1) and (2) for $\eta
\geq1/2$ are optimal, provided that a longstanding minimum output entropy
conjecture is true. Finally, we have generalized the achievable rate regions
in (1) and (2)\ to the case of thermal-noise and amplifying bosonic channel,
with the latter results applying to the Unruh channel studied in previous work.

There are certainly some interesting questions to consider for future work.
First, it would be great to lay out explicit encoding-decoding architectures
that come close to achieving the rate regions given in (1) and (2). Progress
along these lines is in Ref.~\cite{WGTL12,WG11,WG11a,WR12,WR12a}\ for the case of classical or quantum
communication alone, but more generally, there might be some way to leverage these
results for a trade-off coding architecture. It would also be good to
investigate whether the recent bounds derived in Ref.~\cite{KS12b}\ could be
used to determine true outer bounds on the regions given here. Finally, it
would be good to go beyond the single-mode approximation for the Unruh
channel, as recent work has demonstrated that this approximation is not
sufficient for modeling the Unruh effect in general quantum
information-theoretic applications \cite{PhysRevA.82.042332}. Refs.~\cite{DRW12,DDMB12}
identify the communication conditions under which effective single-mode Unruh channels
can be identified.

We thank J.~H.~Shapiro for reminding us of relevant results~\cite{GSE07}
and Nick Menicucci for a useful discussion. MMW
acknowledges the MDEIE\ (Qu\'{e}bec) PSR-SIIRI international collaboration
grant. PH\ acknowledges the hospitality of the Stanford Institute for
Theoretical Physics as well as funding from the Canada Research Chairs
program, the Perimeter Institute, CIFAR, FQRNT's INTRIQ, NSERC, ONR through
grant N000140811249, and QuantumWorks. SG acknowledges the DARPA Information
in a Photon program, contract \#HR0011-10-C-0159.

%\bibliography{Ref}
%\bibliographystyle{apsrev}

\end{document}